\documentclass[11pt]{article}

\usepackage{amsmath, amsfonts, amssymb, amsthm}

\usepackage{pgfplots}

\usepackage{subcaption} 

\usepackage{tikz}

\usepackage{wrapfig} 
\usepackage{algorithm}
\usepackage{algpseudocode}
\usepackage{listings}
\usepackage{xcolor}
\usepackage{tikz}

\usetikzlibrary{arrows.meta, positioning, shapes.geometric}
\usepackage{booktabs} 

\usepackage{amsmath}
\usepackage{amsfonts}
\usepackage{amssymb}
\usepackage{verbatim}
\usepackage[left=2.9cm,right=2.9cm,top=3cm,bottom=3cm]{geometry}

\usepackage[utf8]{inputenc}
\usepackage[T1]{fontenc}

\usepackage{enumerate}

\usepackage{graphicx}
\usepackage{hyperref}

\usepackage[utf8]{inputenc}

\usepackage{booktabs}

\usepackage{algpseudocode}
\usepackage{listings}
\usepackage{xcolor}

\usepackage{amsmath}
\usepackage{amsfonts}
\usepackage{amssymb}
\usepackage{verbatim}
\usepackage[left=2.9cm,right=2.9cm,top=3cm,bottom=3cm]{geometry}

\usepackage[utf8]{inputenc}
\usepackage[T1]{fontenc}

\usepackage{enumerate}

\usepackage{graphicx}

\usepackage{hyperref}

\newcommand{\mP}{{\mathbb P}}

  \definecolor{dgreen}{rgb}{.64,.0,.64}

\newtheorem{theorem}{Theorem}
\newtheorem{lemma}{Lemma}

\newtheorem{corollary}{Corollary}

\newtheorem{remark}{Remark}

\newtheorem{assumption}{Assumption}

\title{Synaptic Classification via Spike-Triggered Extrapolation\thanks{Dedicated to the memory of Antonio Galves (1943--2023), a mentor and collaborator whose brilliance continues to guide this field.}}

\author{
    Emilio De Santis \\
    \small Dipartimento di Matematica, Sapienza Università di Roma \\
    \small Piazzale Aldo Moro, 5, 00185, Rome, Italy \\
 \small \texttt{desantis@mat.uniroma1.it}, \texttt{emilio.desantis@uniroma1.it}
}

\begin{document}

\maketitle

\begin{abstract}

This work introduces a statistical procedure to infer the interaction graph of neuronal networks modeled by Galves-Löcherbach dynamics. The methodology performs bivariate inference, identifying synaptic links from the spike trains of pairs of neurons without observing the rest of the network.

We propose a Macro-Micro Extrapolation algorithm to address data sparsity by inferring interactions in the limit $\Delta \to 0^+$. The core component is a Spike-Triggered Estimator that leverages the local reset property to decouple synaptic jumps from background noise. By employing an adaptive logic that switches between sample averaging and Pyramid Extrapolation, the framework categorizes connections as excitatory, inhibitory, or null.

Numerical simulations demonstrate that the classifier identifies synapses without error across varying noise regimes and complex network topologies, even for observation windows broader than those predicted by the current theoretical bounds.

\medskip

\end{abstract}

\section{Introduction}

The identification of the underlying interaction graph between neurons, based on the recording of spike train time series, remains a fundamental statistical challenge \cite{DeSantis2022,  Duarte2019_Bernoulli, RB2013, Spaziani2025}.  Within the Galves-Löcherbach framework, this inference is technically demanding due to the rapid temporal dependencies of the model, which are often confounded by data sparsity and  background noise. In recent years, synaptic connectivity has been studied primarily through models that integrate diverse observed signals, such as brain rhythms and neuronal spikes \cite{kkk2019, Gerhard2011, Spaziani2025}. Our approach, however, focuses exclusively on spiking data—the fundamental discrete signal typically available in these contexts.

We focus our study on the continuous-time Galves-Löcherbach model \cite{Galves2013, Hodara2017} to present our techniques for identifying synaptic connections. Although we believe our methods are sufficiently general, we apply them here to a single model to avoid technical complexities that would obscure the core methodology. A key strength of our approach is that it identifies synaptic links using only the bivariate spiking history of the two targeted neurons—a minimal information set. Specifically, this work builds upon the research in \cite{DeSantis2022} by introducing a modified, more effective estimator and deriving sharper statistical error bounds for synaptic classification.

A primary difficulty inherent to the methodology in \cite{DeSantis2022} is that the statistical validity of the inference is tied to the choice of an observation window $\Delta$ that is sufficiently small to minimize the approximation bias. This requirement, in turn, leads to the extreme sparsity of the associated counting measure, as the statistics defined in \cite{DeSantis2022} are based on the ratio between sums of Bernoulli random variables whose success parameters are $\Theta(\Delta^2)$. Consequently, since asymptotic unbiasedness of the estimator constrains the analysis to such a micro-scale, an exceptionally long observation time would be required to obtain statistically significant results. 

To overcome the statistical challenges of micro-scale analysis, we introduce a refined estimation framework that improves upon \cite{DeSantis2022} through two distinct innovations.

First, we modify the original estimator by dynamically anchoring the observation windows to the target neuron's spikes. This anchoring represents an intrinsic statistical improvement: by exploiting the local reset property of the dynamics, the system is evaluated from a known state, which reduces the estimator's variance and empirically captures a higher number of informative events, yielding a more robust statistic regardless of any asymptotic procedure.

Second, to bypass the prohibitive observation times required by microscopic sampling, we introduce a Macro-Micro Extrapolation algorithm \cite{Richardson1911, TESTO}. To mathematically enable this approach, we apply a $\Delta^{-1}$ rescaling to our anchored statistic. While dividing by a constant is statistically irrelevant for inference at a single, fixed window size $\Delta$, it is crucial for taking the asymptotic limit. Without it, the statistic trivially vanishes as $\Delta \to 0^+$. 
The rescaling compensates for this degeneracy, ensuring convergence to a well-defined, non-zero value that isolates the synaptic interaction. By evaluating this properly scaled estimator over larger, computationally manageable observation windows, we robustly extrapolate its trend to the micro-scale limit using a novel \textbf{Pyramid Extrapolation} logic. This geometric approach acts as a discrete filter to neutralize network-induced non-linearities, allowing us to reliably classify connections as excitatory, inhibitory, or null, and achieving perfect accuracy even in dense topological structures and high-interference regimes (up to $N=60$).

It is also worth noting that recent extensions of the model integrate more complex neurobiological mechanisms, such as short-term synaptic facilitation \cite{Galves2020}, axonal conduction delays, and membrane potential leakage (see, e.g., \cite{D2026}). The methodologies introduced in this work are built to be effective beyond the specific cases presented here; preliminary tests on the model proposed in \cite{D2026} already suggest that the adaptive extrapolation logic maintains its robustness even in the presence of more complex dynamics. Finally, the mathematical relevance of the Galves-Löcherbach framework extends well beyond neurobiology: it has recently been applied to social networks (see \cite{GalvesLaxa2024, LL2024}) to model fast consensus and metastability in polarized communities, highlighting its broad multidisciplinary appeal.

The article is organized as follows. In Section \ref{Sec:2}, we define the stochastic neural model based on the Galves-Löcherbach framework. Section \ref{Sec:3} introduces the core of our methodological shift: the Spike-Triggered Estimator, highlighting its advantages over traditional fixed-grid sampling. In Section \ref{Sec:4}, we derive the statistical properties of the estimator. 
Section \ref{Sec:5} provides rigorous bounds on the conditional probabilities related to our estimators. We also analytically identify a critical window width ($\Delta_A$) required for single-scale identification, noting that such a resolution is often too fine for practical computational applications. 
Section \ref{Sec:6} is dedicated to sample complexity and the derivation of error bounds. Section \ref{Sec:7} provides numerical validation of the single-scale estimator, confirming the theoretical predictions. Section \ref{Sec:8} introduces the Macro-Micro Strategy and the Pyramid Extrapolation logic to overcome the limitations of the ultra-fine sampling window. Section \ref{Sec:9} evaluates the framework under extreme noise conditions using a minimal architecture. Section \ref{Sec:10} demonstrates its scalability and robustness in fully interactive and layered networks up to $N=60$. Section \ref{Sec:11} summarizes our findings and outlines future research directions.


\section{The Continuous-Time Galves-Löcherbach Model} \label{Sec:2}

We consider a neuronal network modeled as a family of interacting point processes $(N^i, i \in I)$ on a probability space $(\Omega, \mathcal{F}, \mathbb{P})$, where $I$ is a finite set of neurons. Following \cite{Hodara2017}, for each $i \in I$, $N^i$ is a random counting measure on $(\mathbb{R}_+, \mathcal{B}(\mathbb{R}_+))$.

For any interval $A = (s, t] \subset \mathbb{R}_+$, $N^i(A)$ denotes the number of spikes emitted by neuron $i$ in $A$. We assume the processes are simple, such that $\mathbb{P}(N^i(\{t\}) \leq 1 \text{ for all } t \geq 0) = 1$. The sequence of spike times $(T^i_k)_{k \geq 1}$ for neuron $i$ is:
\begin{equation}\label{Tk}
	T^i_k = \inf \{ t > T^i_{k-1} : N^i((T^i_{k-1}, t]) = 1 \}.
\end{equation}
The process $N^i$ is expressed as $N^i(A) = \sum_{k=1}^{\infty} \mathbf{1}_{A}(T^i_k)$. For each neuron $i \in I$, the \emph{firing intensity} $\lambda_i(t)$ is:
\begin{equation} \label{intensity}
	\lambda_i(t) = \phi_i(u_i(t)),
\end{equation}
where $u_i(t)$ is the \emph{membrane potential} and $\phi_i: \mathbb{R} \to \mathbb{R}_+$ is a non-negative, non-decreasing function.

The interaction between neurons is mediated by synaptic weights $w^{j \to i}$. For each $i \in I$, the \emph{pre-synaptic neighborhood} is $\mathcal{V}^i := \{j \in I : w^{j \to i} \neq 0\}$, partitioned into the excitatory subset $\mathcal{V}^i_+ = \{j \in I : w^{j \to i} > 0\}$ and the inhibitory subset $\mathcal{V}^i_- = \{j \in I : w^{j \to i} < 0\}$. We assume a maximum in-degree $d$, such that $|\mathcal{V}^i| \le d$.

The system of membrane potentials $\mathbf{u}(t) = (u_i(t))_{i \in I}$ defines the state of the network. For any $t \ge 0$, the potential of neuron $i$ is determined by the initial state $\mathbf{u}(0)$ and the subsequent spikes:
\begin{equation}\label{membrane_markov}
	u_i(t) = u_i(0)\mathbf{1}_{\{N^i([0, t)) = 0\}} + \sum_{j \in \mathcal{V}^i} w^{j \to i} N^j((L_i(t) \vee 0, t)) ,
\end{equation}
where $L_i(t) = \sup\{s < t : N^i(\{s\}) = 1\}$ is the last spike time of neuron $i$ before $t$. When neuron $i$ spikes, its potential $u_i(t)$ resets to zero, and the system satisfies the Markov property.

The Markovian formulation is advantageous for our \emph{spike-triggered} approach: by anchoring observation windows to the spikes of neuron $i$, we sample the process when the potential $u_i$ is zero.

To formalize the information flow, let $\mathcal{F}_t$ be the $\sigma$-algebra representing the full state of the system up to time $t$, and let $\mathcal{H}_t \subset \mathcal{F}_t$ be the $\sigma$-algebra generated by the observable spike trains.

\begin{remark}  \normalfont
	By the tower property of conditional expectation, any spiking event $E_t$ relates to the observable history as:
	$$ \mathbb{P}(E_t| \mathcal{H}_t) = \mathbb{E}[\mathbb{P}(E_t | \mathcal{F}_t) | \mathcal{H}_t] \,. $$
\end{remark}

Following \cite{DeSantis2022}, we assume the firing intensity functions $\phi_i : \mathbb{R} \to \mathbb{R}_+$ satisfy

\begin{assumption} \label{assu}
The intensity functions $(\phi_i)_{i \in I}$ and the network structure satisfy:
\begin{enumerate}
    \item \textbf{Monotonicity:} $\phi_i$ is non-decreasing for every $i \in I$;
    \item \textbf{Strictly positive lower bound:} $\min_{i \in I} \inf_{u \in \mathbb{R}} \phi_i(u) \geq \alpha > 0$;
    \item \textbf{Uniform upper bound:} $\max_{i \in I} \sup_{u \in \mathbb{R}} \phi_i(u) \leq \beta < +\infty$;
    \item \textbf{Minimum synaptic impact:} $\min_{j \in \mathcal{V}^i} |\phi_i(w^{j \to i}) - \phi_i(0)| \geq \delta > 0$;
    \item \textbf{Bounded in-degree:} $\max_{i \in I} |\mathcal{V}^i | \leq d$.
\end{enumerate}
\end{assumption}

The parameters $\alpha, \beta, \delta \in \mathbb{R}^+$ and $d \in \mathbb{N}$ are assumed to be fixed and known.



\section{Spike-Triggered Sampling Scheme} \label{Sec:3}

In this section, we introduce the \emph{Spike-Triggered Scheme} and define the associated sampling events using the notation $\mathcal{A}^i_k, \mathcal{B}^i_k, \mathcal{C}^{j \to i}_k$, and $\mathcal{D}^{j \to i}_k$. In this framework, the index $k$ refers to two sequences of stopping times, $(\tau^i_k)_{k \in \mathbb{N}}$ and $(\sigma^{j \to i}_k)_{k \in \mathbb{N}}$, which are derived from the firing history of the targeted neurons $i$ and $j$. These stopping times serve as \emph{dynamic triggers} for the observation windows. This approach allows for a recursive definition of the sampling events that explicitly exploits the local reset property of the membrane potential.

\subsection{Baseline Triggered Scheme for $\mathcal{A}^i$ and $\mathcal{B}^i$}

Let $(T^i_j)_{j \ge 1}$ be the ordered sequence of spike times of neuron $i$. We define a sequence of \emph{trigger times} $\{\tau^i_k\}_{k \ge 1}$ and the corresponding sampled events through the following recursive design:

\begin{enumerate}
    \item \textbf{Initialization}: 
    \begin{equation}
        \tau^i_1 := T^i_1.
    \end{equation}
    
    \item \textbf{Definition of Triggered Events}: For each trigger $\tau^i_k$, we define:
    \begin{itemize}
        \item The \textit{reference event} $\mathcal{A}^i_k$, marking the reset of the membrane potential:
        \begin{equation}
            \mathcal{A}^i_k := \{ N^i(\{\tau^i_k\}) = 1 \}.
        \end{equation}
        \item The \textit{bursting event} $\mathcal{B}^i_k(\Delta)$, occurring if neuron $i$ emits at least one subsequent spike within the window $(\tau^i_k, \tau^i_k + \Delta]$:
        \begin{equation}
            \mathcal{B}^i_k(\Delta) := \{ N^i(\tau^i_k, \tau^i_k + \Delta] > 0 \}.
        \end{equation}
    \end{itemize}
    
    \item \textbf{Conditional Recursive Update}: The next trigger time $\tau^i_{k+1}$ is determined by the outcome of $\mathcal{B}^i_k(\Delta)$:
    \begin{itemize}
        \item \textbf{If $\mathcal{B}^i_k(\Delta)$ occurs}: Let $T^i_{k,burst}$ be the first spike in $(\tau^i_k, \tau^i_k + \Delta]$. To preserve the reset condition for the next trial, we skip this spike and set:
        \begin{equation}
            \tau^i_{k+1} := \inf \{ T^i_{\ell} : T^i_{\ell} > T^i_{k,burst} \}.
        \end{equation}
        
        \item \textbf{If $\mathcal{B}^i_k(\Delta)$ does not occur}: The next trigger is the first spike occurring after the current horizon $\Delta$:
        \begin{equation}
            \tau^i_{k+1} := \inf \{ T^i_{\ell} : T^i_{\ell} > \tau^i_k + \Delta \}.
        \end{equation}
    \end{itemize}
\end{enumerate}

\begin{figure}[ht]
	\centering
	\begin{tikzpicture}[>=stealth, scale=1.1]
		\draw[->] (0,0) -- (10,0) node[right] {$t$};
		
		\node at (-1, 1.5) [left] {Neuron $i$};
		\draw[thick] (0,1.5) -- (9.5,1.5);
		
		\draw[blue, ultra thick, ->] (1,1.5) -- (1,2.2) node[above] {$\sigma^{j \to i}_k$};
		\draw[fill=white] (1,1.5) circle (2pt); 
		
		\node at (-1, 0.5) [left] {Neuron $j$};
		\draw[thick] (0,0.5) -- (9.5,0.5);
		\draw[orange, ultra thick, ->] (2.2,0.5) -- (2.2,1.1) node[above] {$T^j_{k,*}$};
		
		\draw[red, ultra thick, ->] (3.8,1.5) -- (3.8,2.2) node[above] {$T^i_{k, \text{react}}$};
		\draw[fill=white] (3.8,1.5) circle (2pt); 
		
		\draw[blue, ultra thick, ->] (7,1.5) -- (7,2.2) node[above] {$\sigma^{j \to i}_{k+1}$};
		\draw[fill=white] (7,1.5) circle (2pt);
		
		\draw[dashed, gray] (1,1.5) -- (1,-0.5);
		\draw[dashed, gray] (3,1.5) -- (3,-0.5);
		\draw[<->] (1,-0.3) -- node[below] {$\Delta$} (3,-0.3);
		
		\draw[dashed, gray] (2.2,1.5) -- (2.2,-0.9);
		\draw[dashed, gray] (4.2,1.5) -- (4.2,-0.9);
		\draw[<->] (2.2,-0.7) -- node[below] {$\Delta$} (4.2,-0.7);
		
		\node at (2, -1.3) {\small $\mathcal{C}^{j \to i}_k$ occurs};
		\node at (4.2, -1.3) {\small $\mathcal{D}^{j \to i}_k$ occurs};
		
	\end{tikzpicture}
	\caption{Spike-Triggered interaction scheme. In this trial, $T^j_{k,*}$ falls within the first window $\Delta$ (event $\mathcal{C}$ occurs). Consequently, a second window is opened at $T^j_{k,*}$, and the arrival of $T^i_{k, \text{react}}$ within this interval implies that event $\mathcal{D}$ occurs.}
	\label{fig:scheme_ST}
\end{figure}

\subsection{Interaction Triggered Scheme for $\mathcal{C}^{j \to i}$ and $\mathcal{D}^{j \to i}$}

We define a sequence of trigger times $(\sigma^{ j \to i}_k)_{k \ge 1}$ for the interaction analysis, where each trigger coincides with a spike of the target neuron $i$.

\begin{enumerate}
	\item \textbf{Initialization and Reset}: The first interaction trigger is $\sigma^{ j \to i}_1 := T^i_1$. At this instant, the membrane potential $u_i$ is reset to zero.
	
	\item \textbf{Causal Triggering}: For each $\sigma^{ j \to i}_k$, we identify the first spike of neuron $j$ after the trigger:
	\begin{equation} \label{Tj*}
		T^j_{k,*} := \inf \{ T^j_{h} : T^j_{h} > \sigma^{ j \to i}_k \}.
	\end{equation}
	The \emph{trigger event} $\mathcal{C}^{j \rightarrow i}_k(\Delta)$ occurs if this spike falls within the window:
	\begin{equation}
		\mathcal{C}^{j \rightarrow i}_k(\Delta) := \{ T^j_{k,*} \le \sigma^{ j \to i}_k + \Delta \}.
	\end{equation}
	
	\item \textbf{Conditional Observation}: If $\mathcal{C}^{j \rightarrow i}_k(\Delta)$ occurs, a secondary window of length $\Delta$ is opened at $T^j_{k,*}$. The \emph{interaction event} occurs if neuron $i$ fires:
	\begin{equation}
		\mathcal{D}^{j \rightarrow i}_k(\Delta) := \mathcal{C}^{j \rightarrow i}_k(\Delta) \cap \{ N^i(T^j_{k,*}, T^j_{k,*} + \Delta] > 0 \}.
	\end{equation}
	
	\item \textbf{Recursive Update}: The next trigger $\sigma^{j \to i}_{k+1}$ is determined as follows:
	\begin{itemize}
		\item \emph{If $\mathcal{C}^{j \rightarrow i}_k(\Delta)$ is not satisfied}: The next trigger is:
		\begin{equation}
			\sigma^{ j \to i}_{k+1} := \inf \{ T^i_{\ell} : T^i_{\ell} > \sigma^{ j \to i}_k + \Delta \}.
		\end{equation}
		
		\item \emph{If $\mathcal{C}^{j \rightarrow i}_k(\Delta)$ is satisfied}: 
		\begin{itemize}
			\item \emph{Case Success ($\mathcal{D}^{j \rightarrow i}_k(\Delta)$ occurs):} Let $T^i_{k, \text{react}} := \inf \{ T^i_{\ell} : T^i_{\ell} > T^j_{k,*} \}$. To maintain the reset condition, we set:
			\begin{equation}
				\sigma^{ j \to i}_{k+1} := \inf \{ T^i_{\ell} : T^i_{\ell} > T^i_{k, \text{react}} \}.
			\end{equation}
			\item \emph{Case Failure ($\mathcal{D}^{j \rightarrow i}_k(\Delta)$ does not occur):} The next trigger is:
			\begin{equation}
				\sigma^{ j \to i}_{k+1} := \inf \{ T^i_{\ell} : T^i_{\ell} > T^j_{k,*} + \Delta \}.
			\end{equation}
		\end{itemize}
	\end{itemize}
\end{enumerate}

The sequence $(\sigma^{j \to i}_k)_{k \ge 1}$ consists of stopping times with respect to the filtration $(\mathcal{F}_t)_{t \ge 0}$. By construction, each trigger coincides with a spike of neuron $i$, ensuring the potential $u_i(t)$ is reset at the start of each trial. This decoupling provides the basis for applying concentration inequalities to the Spike-Triggered Classifier.


\section{The Spike-Triggered Estimator $\hat{\mathcal{G}}^{j \to i}(\Delta)$} \label{Sec:4}

The sampling protocol introduced in Section \ref{Sec:3} allows for the construction of a statistic designed to isolate the functional influence between neurons. For a pair of neurons $(j, i)$, we define the \emph{spike-triggered estimator} as:
\begin{equation} \label{def_hat_G_cal}
	\hat{\mathcal{G}}^{j \to i}_{m_0, m_1}(\Delta) := \frac{1}{\Delta \delta} 
	\left( \frac{\sum_{k=1}^{m_1} \mathbf{1}_{\mathcal{D}^{j \to i}_k(\Delta)}}{\sum_{k=1}^{m_1} \mathbf{1}_{\mathcal{C}^{j \to i}_k(\Delta)}} - 
	\frac{\sum_{k=1}^{m_0} \mathbf{1}_{\mathcal{B}^i_k(\Delta)}}{m_0} \right).
\end{equation}

\begin{remark} \normalfont
	It is worth noting that the statistic originally introduced in \cite{DeSantis2022} was based on a fixed observation grid. The spike-triggered estimator $\hat{\mathcal{G}}^{j \to i}$ in \eqref{def_hat_G_cal} represents a novel construction that leverages dynamic anchors to evaluate the system from a known state (the local reset). Furthermore, the inclusion of the normalization factor $(\Delta \delta)^{-1}$ ensures that the statistic remains bounded away from zero as $\Delta \to 0$ in the presence of direct synaptic interaction. Specifically, since the baseline spiking probability in a window of length $\Delta$ is $O(\Delta)$, this scaling factor allows the estimator to converge to $(\phi_i(w^{j \to i}) - \phi_i(0))/\delta$, as will be proved in the following section.
\end{remark}

\subsection{Conditional Probabilities and Ergodic Means}

We define the theoretical quantity for a representative trial ($k=1$), which depends on the realization of the history up to the trigger times. For the spike-triggered approach, this is given by:
\begin{equation} \label{G_cal_theo_k1} 
	\mathcal{G}^{j \to i}(\Delta, \mathcal{F}_{\sigma^{j \to i}_1}, \mathcal{F}'_{\tau^{i}_1}) := \frac{1}{\Delta \delta} \left[ \frac{\mathbb{P}(\mathcal{D}^{j \to i}_1(\Delta) \mid \mathcal{F}_{\sigma^{j \to i}_1})}{\mathbb{P}(\mathcal{C}^{j \to i}_1(\Delta) \mid \mathcal{F}_{\sigma^{j \to i}_1})} - \mathbb{P}(\mathcal{B}^i_1(\Delta) \mid \mathcal{F}'_{\tau^{i}_1}) \right].
\end{equation}
In \eqref{G_cal_theo_k1}, the denominator of the baseline term is 1 since $\mathbb{P}(\mathcal{A}^i_1 \mid \mathcal{F}'_{\tau^{i}_1}) = 1$ by construction.  

The theoretical quantity defined in \eqref{G_cal_theo_k1} is a random variable, as it is defined through conditional probabilities with respect to the stopped $\sigma$-algebras $\mathcal{F}_{\sigma^{j \to i}_1}$ and $\mathcal{F}'_{\tau^{i}_1}$. These represent the information flow stopped exactly at the reset times of the target neuron $i$ for an interaction trial and a baseline trial, respectively. It is important to emphasize that $\mathcal{F}_{\sigma^{j \to i}_1}$ and $\mathcal{F}'_{\tau^{i}_1}$ represent distinct realizations of the past activity; they do not necessarily refer to the same point in absolute time nor to the same sequence of events. The random nature of this quantity arises directly from its dependence on the specific history of the system preceding each trial. 

Let $\nu$ be the equilibrium measure of the system under spike-triggered sampling. The expected value is:
\begin{equation} \label{E_G_adapt}
	\mathcal{G}^{j \to i}(\Delta) := \mathbb{E}_\nu [ \mathcal{G}^{j \to i}(\Delta, \mathcal{F}_{\sigma^{j \to i}_1}, \mathcal{F}'_{\tau^{i}_1}) ]. 
\end{equation}

Under Assumption \ref{assu}, the process is ergodic \cite{Hodara2017}. Consequently, the estimator converges almost surely to its ergodic mean:
\begin{equation}\label{ergodico}
	\hat{\mathcal{G}}_{m_0, m_1}^{j \to i}(\Delta) \xrightarrow[m_0, m_1 \to \infty]{a.s.} \mathcal{G}^{j \to i}(\Delta).
\end{equation}


\section{Probabilistic Bounds} \label{Sec:5}

We study the range of values that the conditional probability $\mathbb{P}(\mathcal{B}^i_k(\Delta) \mid \mathcal{F}_{\tau^i_k})$ can take, based on the history of the process up to the stopping time $\tau^i_k$. 

\begin{lemma}[Baseline Probability Estimates] \label{L1}
	Under Assumption 1, for any trigger time $\tau_{k}^{i}$, the random variable $\mathbb{P}(\mathcal{B}^i_k(\Delta) \mid \mathcal{F}_{\tau^i_k})$ satisfies the following inequalities almost surely:
	\begin{equation}\label{formula28}
		(1-e^{-\phi_i(0)\Delta})e^{-\beta |\mathcal{V}^i|\Delta} \le \mathbb{P}(\mathcal{B}^i_k(\Delta) \mid \mathcal{F}_{\tau^i_k}) \leq (1-e^{-\phi_i(0)\Delta}) e^{- \alpha |\mathcal{V}^i| \Delta} + (1-e^{-\beta \Delta})(1-e^{-\beta |\mathcal{V}^i| \Delta}) .
	\end{equation}
\end{lemma}

\begin{proof}
	The event $\mathcal{B}_{k}^{i}(\Delta)$ occurs if neuron $i$ emits at least one spike in $(\tau_k^i, \tau_k^i + \Delta]$. Since we condition on $\mathcal{F}_{\tau^i_k}$, all subsequent inequalities hold almost surely. By the reset property, $u_i(\tau_k^i) = 0$ a.s., implying the initial firing intensity is $\lambda_i(\tau_k^i) = \phi_i(0)$.
	
	To derive the bounds, we introduce two auxiliary Poisson processes, $P_*$ and $P^*$, with constant intensities $\alpha |\mathcal{V}^i|$ and $\beta |\mathcal{V}^i|$, respectively. These processes are constructed to be independent of the history $(\mathcal{F}_t)_{t \ge 0}$. Let $N^{\mathcal{V}^i}$ be the point process of the pre-synaptic neighborhood $\mathcal{V}^i$. By Assumption 1, $P_*$ and $P^*$ stochastically dominate $N^{\mathcal{V}^i}$ from below and above, respectively.
	
	For the \emph{lower bound}, we consider the case where no pre-synaptic spikes occur:
	\begin{align} \label{lower}
		\mathbb{P}(\mathcal{B}^i_k(\Delta) \mid \mathcal{F}_{\tau^i_k}) &\geq \mathbb{P}(\mathcal{B}^i_k(\Delta) \cap \{N^{\mathcal{V}^i}(\tau_k^i, \tau_k^i + \Delta] = 0\} \mid \mathcal{F}_{\tau^i_k}) \nonumber \\
		&\geq \mathbb{P}(\mathcal{B}^i_k(\Delta) \cap \{P^*(\tau_k^i, \tau_k^i + \Delta] = 0\} \mid \mathcal{F}_{\tau^i_k}) \nonumber \\
		&= (1-e^{-\phi_i(0)\Delta})e^{-\beta |\mathcal{V}^i|\Delta} ,
	\end{align}
	where the equality follows from the independence of $P^*$.
	
	For the \emph{upper bound}, we partition $\mathcal{B}^i_k(\Delta)$ based on the neighborhood activity. By replacing $N^{\mathcal{V}^i}$ with the independent processes $P_*$ and $P^*$, we obtain:
	\begin{align}
		\mathbb{P}(\mathcal{B}^i_k(\Delta) \mid \mathcal{F}_{\tau^i_k}) &\leq \mathbb{P}(\mathcal{B}^i_k(\Delta) \cap \{N^{\mathcal{V}^i} = 0\} \mid \mathcal{F}_{\tau^i_k}) + \mathbb{P}(\mathcal{B}^i_k(\Delta) \cap \{N^{\mathcal{V}^i} > 0\} \mid \mathcal{F}_{\tau^i_k}) \nonumber \\
		&\leq \mathbb{P}(\mathcal{B}^i_k(\Delta) \cap \{P_* = 0\} \mid \mathcal{F}_{\tau^i_k}) + \mathbb{P}(\mathcal{B}^i_k(\Delta) \cap \{P^* > 0\} \mid \mathcal{F}_{\tau^i_k}) \nonumber \\
		&\leq (1-e^{-\phi_i(0)\Delta}) e^{- \alpha |\mathcal{V}^i| \Delta} + (1-e^{-\beta \Delta})(1-e^{-\beta |\mathcal{V}^i| \Delta}). \label{finaleL1}
	\end{align}
	In \eqref{finaleL1}, the second term is bounded by assuming the maximum possible intensity $\beta$ for neuron $i$ throughout $\Delta$.
\end{proof}

Using the Taylor expansion with a Lagrange remainder one has 
\begin{equation}\label{uso_ripetuto}
1 - x \le e^{-x} \le 1 - x + \frac{x^2}{2},  \text{ and }
e^{-x} \ge 1-x
\end{equation} for $x \ge 0$, and assuming $|\mathcal{V}^i| \leq d$ and $\phi_i(0) \le \beta$, Lemma \ref{L1} yields the following uniform bounds.

\begin{corollary}\label{C1}
	Under the assumptions of Lemma \ref{L1}, for any $\Delta > 0$:
	\begin{equation} \label{sandwich_uniform}
		\phi_i(0)\Delta - \left(d + \frac{1}{2}\right) \beta^2 \Delta^2 \le \mathbb{P}(\mathcal{B}^i_k(\Delta) \mid \mathcal{F}_{\tau^i_k}) \le \phi_i(0)\Delta + d \beta^2 \Delta^2 .
	\end{equation}
\end{corollary}

\begin{proof}
	For the lower bound, from \eqref{lower} we have:
	\begin{align*}
		\mathbb{P}(\mathcal{B}^i_k(\Delta) \mid \mathcal{F}_{\tau^i_k}) &\ge (1 - e^{-\phi_i(0)\Delta}) e^{-\beta |\mathcal{V}^i|\Delta} \\
		&\ge \left(\phi_i(0)\Delta - \frac{\phi_i(0)^2 \Delta^2}{2}\right) (1 - \beta |\mathcal{V}^i|\Delta) \\
		&= \phi_i(0)\Delta - \left( \beta |\mathcal{V}^i| \phi_i(0) + \frac{\phi_i(0)^2}{2} \right) \Delta^2 + \frac{\phi_i(0)^2 \beta |\mathcal{V}^i| \Delta^3}{2} \\
		&\ge \phi_i(0)\Delta - \left(d \beta^2 + \frac{\beta^2}{2}\right) \Delta^2,
	\end{align*}
	where we used $1-e^{-x} \ge x - x^2/2$ and $e^{-x} \ge 1-x$, and dropped the positive cubic term.
	
	For the upper bound, from \eqref{finaleL1} we observe that $e^{-\alpha |\mathcal{V}^i| \Delta} \le 1$. Using the inequality $1-e^{-x} \le x$ for $x \ge 0$, we obtain:
	\begin{align*}
		\mathbb{P}(\mathcal{B}^i_k(\Delta) \mid \mathcal{F}_{\tau^i_k}) &\le (1-e^{-\phi_i(0)\Delta}) \cdot 1 + (1-e^{-\beta \Delta})(1-e^{-\beta |\mathcal{V}^i| \Delta}) \\
		&\le \phi_i(0)\Delta + (\beta \Delta)(\beta |\mathcal{V}^i| \Delta) \\
		&\le \phi_i(0)\Delta + d \beta^2 \Delta^2,
	\end{align*}
	where we substituted $|\mathcal{V}^i| \le d$.
\end{proof}

The following Lemma \ref{L2} describes how the occurrence of a spike in neuron $j$ affects the spiking probability of neuron $i$ in the subsequent window $\Delta$.

\begin{lemma}[Probability Ratios for Connectivity] \label{L2} 
	Let $\Delta >0$. For any target neuron $i$ and candidate pre-synaptic neuron $j$, the conditional probability of an interaction success satisfies the following bounds $\mathbb{P}$-a.s.:
	\begin{enumerate}
		\item \textbf{No Connection ($j \notin \mathcal{V}^i$):}
		\begin{equation} \label{L2_no}
			(1 - e^{-\phi_i(0)\Delta}) e^{-2 \beta |\mathcal{V}^i| \Delta} \leq \frac{\mathbb{P}(\mathcal{D}^{j \to i}_k(\Delta) \mid \mathcal{F}_{\sigma^{j\to i}_k})}{\mathbb{P}(\mathcal{C}^{j \to i}_k(\Delta) \mid \mathcal{F}_{\sigma^{j\to i}_k})} \leq (1 - e^{-\phi_i(0)\Delta}) + (1-e^{-\beta  \Delta} )(1 - e^{-2\beta |\mathcal{V}^i| \Delta}) 
		\end{equation}
		\item \textbf{Excitatory Connection ($j \in \mathcal{V}^i_+$):}
		\begin{equation} \label{L2_exc}
			\frac{\mathbb{P}(\mathcal{D}^{j \to i}_k(\Delta) \mid \mathcal{F}_{\sigma^{j\to i}_k})}{\mathbb{P}(\mathcal{C}^{j \to i}_k(\Delta) \mid \mathcal{F}_{\sigma^{j\to i}_k})} \ge (1 - e^{-(\phi_i(0)+ \delta )\Delta}) e^{-2\beta (|\mathcal{V}^i| -1) \Delta}
		\end{equation}
		\item \textbf{Inhibitory Connection ($j \in \mathcal{V}^i_-$):}
		\begin{equation} \label{L2_inh}
			\frac{\mathbb{P}(\mathcal{D}^{j \to i}_k(\Delta) \mid \mathcal{F}_{\sigma^{j\to i}_k})}{\mathbb{P}(\mathcal{C}^{j \to i}_k(\Delta) \mid \mathcal{F}_{\sigma^{j\to i}_k})} \leq (1 - e^{-(\phi_i(0)- \delta)\Delta}) + (1-e^{-\beta  \Delta} )(1 - e^{-2\beta (|\mathcal{V}^i|-1) \Delta}) 
		\end{equation}
	\end{enumerate}
\end{lemma}

\begin{proof}
	Since $\mathcal{D}^{j \to i}_k(\Delta) \subseteq \mathcal{C}^{j \to i}_k(\Delta)$, the ratio reduces to the conditional probability:
	\begin{equation} 
		\frac{\mathbb{P}(\mathcal{D}^{j \to i}_k(\Delta) \mid \mathcal{F}_{\sigma^{j\to i}_k})}{\mathbb{P}(\mathcal{C}^{j \to i}_k(\Delta) \mid \mathcal{F}_{\sigma^{j\to i}_k})} = \mathbb{P}(\mathcal{D}^{j \to i}_k(\Delta) \mid \mathcal{C}^{j \to i}_k(\Delta), \mathcal{F}_{\sigma^{j\to i}_k}) . 
	\end{equation}
	Recalling the definition in \eqref{Tj*}, let $T^j_{k,*} = t_* \in (\sigma^{j\to i}_k, \sigma^{j\to i}_k + \Delta)$ be the time of the trigger spike. To establish uniform bounds, we consider the dominating Poisson processes $P^*$ (with intensity $|\mathcal{V}^i | \beta$) and $P^{**}$ (with intensity $\beta$), constructed on an enlarged probability space to be independent of each other and of the filtration $(\mathcal{F}_t)_{t \ge 0}$. These processes provide uniform stochastic dominance for the neighborhood activity $N^{\mathcal{V}^i}$ and the neuron $i$, respectively.
	
	\medskip
	\noindent \textbf{Upper bound ($j \notin \mathcal{V}^i$).} 
	We partition the firing event of neuron $i$ based on the activity of its pre-synaptic neighborhood $\mathcal{V}^i$ in the time interval $\mathcal{T} = (\sigma^{j\to i}_k, t_* + \Delta]$. Let $\mathcal{E}$ be the event $\{N^{\mathcal{V}^i}(\mathcal{T}) = 0\}$. We decompose the probability as follows:
	\begin{align}
		\mathbb{P}(\mathcal{D}^{j \to i}_k \mid T^j_{k,*} = t_*, \mathcal{F}_{\sigma^{j\to i}_k}) &= \mathbb{P}(\mathcal{D}^{j \to i}_k \cap \mathcal{E} \mid \mathcal{F}_{\sigma^{j\to i}_k}) + \mathbb{P}(\mathcal{D}^{j \to i}_k \cap \mathcal{E}^c \mid \mathcal{F}_{\sigma^{j\to i}_k}) \nonumber \\
		&\le \mathbb{P}(\mathcal{D}^{j \to i}_k \mid \mathcal{E}, \mathcal{F}_{\sigma^{j\to i}_k}) + \mathbb{P}(\{P^{**}(0,  \Delta] \ge 1\} \cap \{P^*(0,2 \Delta) \ge 1\} ) \nonumber \\
		&\le (1 - e^{-\phi_i(0)\Delta}) + (1 - e^{-\beta \Delta})(1 - e^{-2|\mathcal{V}^i|\beta \Delta}). \label{upper_rigorosa}
	\end{align}
	In the first term of \eqref{upper_rigorosa}, the intensity of neuron $i$ on the event $\mathcal{E}$ is exactly $\phi_i(0)$ because $u_i(\sigma^{j\to i}_k)=0$ and no pre-synaptic spikes from $\mathcal{V}^i$ occur in $\mathcal{T}$. In the second term, we utilize the independence of the auxiliary processes $P^*$ and $P^{**}$ to obtain a product of probabilities, ensuring an error term of order $O(\Delta^2)$.
	
	\medskip
	\noindent \textbf{Lower bound ($j \notin \mathcal{V}^i$).} 
	Focusing on the scenario where the neighborhood remains silent, and noting that the interval $\mathcal{T}$ has a maximum length of $2\Delta$, we have:
	\begin{align} 
		\mathbb{P}(\mathcal{D}^{j \to i}_k \mid T^j_{k,*} = t_*, \mathcal{F}_{\sigma^{j\to i}_k}) &\geq \mathbb{P}(\mathcal{D}^{j \to i}_k \cap \mathcal{E} \mid \mathcal{F}_{\sigma^{j\to i}_k}) \nonumber \\
		&= \mathbb{P}(\mathcal{D}^{j \to i}_k \mid \mathcal{E}, \mathcal{F}_{\sigma^{j\to i}_k}) \mathbb{P}(\mathcal{E} \mid \mathcal{F}_{\sigma^{j\to i}_k}) \nonumber \\ 
	&\geq  \mathbb{P}(\mathcal{D}^{j \to i}_k \mid \mathcal{E}, \mathcal{F}_{\sigma^{j\to i}_k}) \mathbb{P}(P^* (0, 2\Delta) =0) 
                                \nonumber \\ 
		&\geq (1 - e^{-\phi_i(0)\Delta}) e^{-2   |\mathcal{V}^i|  \beta \Delta}, \label{lower_rigorosa}
	\end{align}
	where we used the stochastic dominance of $P^*$ over $N^{\mathcal{V}^i}$.

\medskip
\noindent \textbf{Excitatory Connection ($j \in \mathcal{V}^i_+$).} 
The trigger spike at $t_*$ induces an increment in the membrane potential of neuron $i$. Let us consider the neighborhood of $i$ excluding the trigger neuron $j$, i.e., $\mathcal{V}^i \setminus \{j\}$, and let $\mathcal{E}^{\setminus j} = \{N^{\mathcal{V}^i \setminus \{j\}}(\mathcal{T}) = 0\}$ be the event that this reduced neighborhood remains silent during $\mathcal{T} = (\sigma^{j\to i}_k, t_* + \Delta]$. 
On the event $\mathcal{E}^{\setminus j}$, the potential $u_i$ is $0$ for $t \in (\sigma^{j\to i}_k, t_*]$ and, at the time of the trigger spike, it jumps to $u_i(t_*^+) = w^{j \to i}$. Since $w^{j \to i} > 0$, any subsequent spikes from $j$ within $(t_*, t_* + \Delta]$ would further increment the potential by $w^{j \to i}$ each time. By the non-decreasing property of $\phi_i$ (Assumption \ref{assu}), the intensity satisfies $\lambda_i(t) = \phi_i(u_i(t)) \ge \phi_i(w^{j \to i}) \ge \phi_i(0) + \delta$ for all $t \in (t_*, t_* + \Delta]$ $\mathbb{P}$-a.s. on $\mathcal{E}^{\setminus j}$. Following the logic in \eqref{lower_rigorosa}, we have
\begin{align} 
    \mathbb{P}(\mathcal{D}^{j \to i}_k \mid T^j_{k,*} = t_*, \mathcal{F}_{\sigma^{j\to i}_k}) &\ge \mathbb{P}(\mathcal{D}^{j \to i}_k \cap \mathcal{E}^{\setminus j} \mid \mathcal{F}_{\sigma^{j\to i}_k}) \nonumber \\
    &\ge (1 - e^{-(\phi_i(0) + \delta)\Delta}) e^{-2(|\mathcal{V}^i|-1)\beta \Delta}, \label{lower_exc_rigorosa}
\end{align}
where the exponential term accounts for the stochastic dominance of $P^*$ over the $|\mathcal{V}^i|-1$ remaining pre-synaptic neurons.

\medskip
\noindent \textbf{Inhibitory Connection ($j \in \mathcal{V}^i_-$).} 
Similarly, for an inhibitory connection, the trigger spike at $t_*$ induces an immediate decrement in the membrane potential, such that $u_i(t_*^+) = w^{j \to i} < 0$. On the event $\mathcal{E}^{\setminus j}$, any subsequent spikes from $j$ within $(t_*, t_* + \Delta]$ would further decrease the potential (making it more negative). By the non-decreasing property of $\phi_i$ (Assumption \ref{assu}), the intensity for $t \in (t_*, t_* + \Delta]$ is bounded from above by $\phi_i(w^{j \to i}) \le \phi_i(0) - \delta$. Thus, the upper bound \eqref{upper_rigorosa} is modified as:
\begin{align}
    \mathbb{P}(\mathcal{D}^{j \to i}_k \mid T^j_{k,*} = t_*, \mathcal{F}_{\sigma^{j\to i}_k}) &\le (1 - e^{-(\phi_i(0) - \delta)\Delta}) + (1 - e^{-\beta \Delta})(1 - e^{-2(|\mathcal{V}^i|-1)\beta \Delta}). \label{upper_inh_rigorosa}
\end{align}
In both cases, the terms involving $(|\mathcal{V}^i|-1)$ correctly reflect that the trigger neuron $j$ is excluded from the silent neighborhood event $\mathcal{E}^{\setminus j}$, as its activity is already accounted for by the conditioning on $T^j_{k,*}$ and the subsequent monotonicity arguments.
\end{proof}

\subsection{Explicit Linear-Quadratic Bounds}

Following the results established in Lemma~\ref{L2}, we simplify the exponential bounds into a more manageable linear-quadratic form. By utilizing \eqref{uso_ripetuto} and $|\mathcal{V}^i| \leq d$, we obtain 

\begin{corollary} \label{cor:explicit_bounds}
Under the assumptions of Lemma \ref{L2}, the ratio of conditional probabilities for any pair of neurons $(j, i)$ satisfies the following explicit bounds almost surely ($P$-a.s.):

\begin{enumerate}
    \item \textbf{No Connection ($j \notin \mathcal{V}^i$):} 
    \begin{equation}
        \phi_i(0)\Delta -  \left( 2d + \frac{1}{2} \right) \beta^2 \Delta^2 \le \frac{\mP(\mathcal{D}_{k}^{j\to i}(\Delta)|\mathcal{F}_{\sigma^{j \to i}_k})}{\mP(\mathcal{C}_{k}^{j\rightarrow i}(\Delta)|\mathcal{F}_{\sigma^{j \to i}_k})} \le \phi_i(0)\Delta + 2d\beta^2 \Delta^2
    \end{equation}

    \item \textbf{Excitatory Connection ($j \in \mathcal{V}_+^i$):}
\begin{equation} \label{eq:44}
\frac{\mP(\mathcal{D}{k}^{j\rightarrow i}(\Delta)|\mathcal{F}{\sigma^{j \to i}_k})}{\mP(\mathcal{C}{k}^{j\rightarrow i}(\Delta)|\mathcal{F}_{\sigma^{j \to i}_k})} \ge (\phi_i(0) + \delta)\Delta - \left( 2d - 1.5 \right) \beta^2 \Delta^2
\end{equation}

    \item \textbf{Inhibitory Connection ($j \in \mathcal{V}_-^i$):}
    \begin{equation}
        \frac{\mP(\mathcal{D}_{k}^{j\rightarrow i}(\Delta)|\mathcal{F}_{\sigma^{j \to i}_k})}{\mP(\mathcal{C}_{k}^{j\rightarrow i}(\Delta)|\mathcal{F}_{\sigma^{j \to i}_k})} \le (\phi_i(0) - \delta)\Delta + 2\left  (d-1 \right )\beta^2 \Delta^2
    \end{equation}
\end{enumerate}
\end{corollary}
\begin{proof}
The result follows by applying the standard exponential inequalities \eqref{uso_ripetuto} to the bounds established in Lemma \ref{L2}, following the same algebraic logic detailed in the proof of Corollary \ref{C1}. Specifically, the lower bounds are obtained by a second-order Taylor expansion with Lagrange remainder, substituting the uniform bounds $\lambda \le \beta$ and $\gamma \le 2d\beta$ (or $\gamma \le 2(d-1)\beta$ where appropriate) and dropping the positive cubic terms. 
\end{proof}

%
%
%
%

\subsection{Rigorous Bounds for $\mathcal{G}^{j \to i}(\Delta)$}

By combining the results of Corollary~\ref{C1} and Corollary~\ref{cor:explicit_bounds}, we derive the formal bounds for the expected quantity $\mathcal{G}^{j \to i}(\Delta)$ defined in \eqref{E_G_adapt}. This quantity isolates the synaptic influence by centering the interaction term around the baseline firing rate and normalizing it by the parameters $\delta$ and $\Delta$.

To ensure that the identification of the connectivity is not affected by finite-window effects, we introduce the operational condition:
\begin{equation}\label{regime}
\Delta < \Delta_{max} := \frac{\delta}{(6d + 1)\beta^2}.
\end{equation}
Under this constraint, the gap between the probability ratios for different connection types remains strictly positive and proportional to $\delta\Delta$. This ensures that the quadratic error does not obscure the excitatory or inhibitory signal.

\begin{theorem}[Stability of Interaction Estimators] \label{TEO-BOUND}
Under  \eqref{regime}, the quantity $\mathcal{G}^{j \to i}(\Delta)$ satisfies the following inequalities:

\begin{enumerate}
    \item \textbf{No Connection ($j \notin \mathcal{V}^i$):}
    \begin{equation}
        - \frac{ (3d + 0.5) \beta^2 }{\delta} \Delta \le \mathcal{G}^{j \to i}(\Delta) \le \frac{ (3d + 0.5) \beta^2 }{\delta} \Delta.
    \end{equation}

    \item \textbf{Excitatory Connection ($j \in \mathcal{V}_{+}^i$):}
    \begin{equation}
        \mathcal{G}^{j \to i}(\Delta) \ge 1 - \frac{ (3d - 1.5) \beta^2}{\delta} \Delta.
    \end{equation}

    \item \textbf{Inhibitory Connection ($j \in \mathcal{V}_{-}^i$):}
    \begin{equation}
        \mathcal{G}^{j \to i}(\Delta) \le -1 + \frac{(3d - 1.5) \beta^2}{\delta} \Delta.
    \end{equation}
\end{enumerate}
\end{theorem}

\begin{proof}
The proof follows from the direct substitution of the second-order expansions established in Corollaries \ref{C1} and \ref{cor:explicit_bounds} into the definition of the expected contrast $\mathcal{G}^{j \to i}(\Delta)$. Since the bounds in the aforementioned corollaries hold almost surely for the random variables $G^{j \to i}(\Delta, \mathcal{F}, \mathcal{F}')$, they are preserved under the expectation operator, thereby providing deterministic bounds for $\mathcal{G}^{j \to i}(\Delta)$.

For all three cases, the fundamental mechanism is the cancellation of the first-order terms proportional to the baseline intensity $\phi_i(0)$. Specifically, when subtracting the baseline ratio from the interaction ratio, the terms of order $O(\Delta)$ involving $\phi_i(0)$ eliminate each other, leaving only the synaptic signal and the accumulated second-order errors.

In the \textbf{No Connection} case, the signal is null. The resulting bounds emerge from the worst-case combination of the asymmetric error terms from both corollaries. The symmetry of the final bound, $\pm(3d + 0.5)$, arises because the lower bound of the interaction term combined with the upper bound of the baseline term yields the same magnitude as the opposite combination (i.e., $(2d+0.5)+d = 2d+(d+0.5) = 3d+0.5$).

For \textbf{Excitatory and Inhibitory} connections, the centering operation isolates the leading term $\pm \delta \Delta$. After dividing by the normalization factor $\delta \Delta$, this leading term yields a base value of $1$ (or $-1$, respectively). The remaining second-order terms, once divided by $\delta \Delta$, become linear in $\Delta$ with coefficients $(3d - 1.5)$. Under the operational condition \eqref{regime}, these fluctuations are strictly dominated by the unit signal, ensuring that $\mathcal{G}^{j \to i}(\Delta) \ge 1 - \frac{(3d-1.5)\beta^2\Delta}{\delta}$ in the excitatory case and $\mathcal{G}^{j \to i}(\Delta) \le -1 + \frac{(3d-1.5)\beta^2\Delta}{\delta}$ in the inhibitory one.
\end{proof}

%
%
%
%

These results allow us to establish a statistical criterion for the identification of excitatory, inhibitory, and null connections.

\begin{corollary}[Topology Identification] \label{cor:identification}
For an observation window $\Delta_A = \frac{\delta}{9d\beta^2}$, the expected contrast $\mathcal{G}^{j \to i}(\Delta_A)$ provides an explicit separation for the network topology identification:

\begin{enumerate}
    \item \textbf{No Connection ($j \notin \mathcal{V}^i$):} 
    \begin{equation}
        |\mathcal{G}^{j \to i}(\Delta_A)| \le \frac{1}{3} + \frac{1}{18d}.
    \end{equation}

    \item \textbf{Excitatory Connection ($j \in \mathcal{V}_{+}^i$):} 
    \begin{equation}
        \mathcal{G}^{j \to i}(\Delta_A) \ge \frac{2}{3} + \frac{1}{6d}.
    \end{equation}

    \item \textbf{Inhibitory Connection ($j \in \mathcal{V}_{-}^i$):} 
    \begin{equation}
        \mathcal{G}^{j \to i}(\Delta_A) \le -\frac{2}{3} - \frac{1}{6d}.
    \end{equation}
\end{enumerate}
\end{corollary}

\begin{remark} \normalfont
The thresholds established in Corollary \ref{cor:identification} define clear safety margins for connectivity inference. In the limit of high regularity or large synaptic density ($d \gg 1$), these bounds converge to the simplified decision regions:
\begin{itemize}
    \item \textbf{Null:} $[-1/3, 1/3]$
    \item \textbf{Excitatory:} $[2/3, +\infty)$
    \item \textbf{Inhibitory:} $(-\infty, -2/3]$
\end{itemize}
This asymptotic behavior ensures a "gap" of size at least $1/3$ between the different classes.
\end{remark}

\medskip

\subsection{Comparison with Previous Identification Bounds}

We compare the analytical bounds established in Corollaries \ref{C1} and \ref{cor:explicit_bounds} with the results presented in Lemma 3 of \cite{DeSantis2022}. The estimates derived in this work provide a sharper characterization of the interaction function $\mathcal{G}^{j \to i}(\Delta)$ for the following reasons.

\paragraph{1. Estimator Construction.}
The estimator in \cite{DeSantis2022} is defined on a fixed-time grid, whereas the current estimator is constructed using the stopping times $(\sigma^{j \to i}_k, \tau^i_k)$. By the reset property, this construction fixes the initial intensity at $\lambda_i = \phi_i(0)$, removing the dependence on the history of the process for the target neuron $i$. This leads to sharper analytical bounds and more stable statistics compared to the fixed-grid approach.

\paragraph{2. Analytical Refinement of the Bounds.}
In Lemma 3 of \cite{DeSantis2022}, the error estimates were dependent on the leak rate $\alpha$. The refined analytical treatment in the current work removes the dependence on $\alpha$ from the denominators, providing a more direct representation of the synaptic interaction. 

Comparing the error constants $C_{2022}$ with the constants $C_{2026}$ from Corollaries \ref{C1} and \ref{cor:explicit_bounds}, we observe a systematic reduction in the error terms. Even assuming $\phi_i(0) = \alpha$, which yields the most favorable estimates for the previous framework, the bounds compare as follows:

\begin{enumerate}
    \item \textbf{Base Firing Ratio, $\mathbb{P}(\mathcal{B}^i_k(\Delta) \mid \mathcal{F}_{\tau^i_k}) $:}
    \begin{itemize}
        \item Lower Bound: $C_{2026} = (d + 0.5)\beta^2$ vs $C_{2022} = 3d\beta^2$
        \item Upper Bound: $C_{2026} = d\beta^2$ vs $C_{2022} \ge 4d\beta^3/\alpha$
    \end{itemize}
    \item \textbf{Probability Ratio $\mathbb{P}(\mathcal{D}^{j \to i}_k(\Delta) \mid \mathcal{C}^{j \to i}_k(\Delta), \mathcal{F}_{\sigma^i_k}) $:} 
    \begin{itemize}
        \item \textbf{Excitatory Case ($j \in \mathcal{V}_+^i$):}
        \begin{itemize}
            \item Lower Bound: $C_{2026} = (2d - 1.5)\beta^2$ vs $C_{2022} \ge 5d\beta^3/\alpha$
        \end{itemize}

        \item \textbf{Inhibitory Case ($j \in \mathcal{V}_-^i$):}
        \begin{itemize}
            \item Upper Bound: $C_{2026} = (2d - 2)\beta^2$ vs $C_{2022} \ge 5d\beta^4/\alpha^2$
        \end{itemize}

        \item \textbf{Null Case ($j \notin \mathcal{V}^i$):}
        \begin{itemize}
            \item Lower Bound: $C_{2026} = (2d + 0.5)\beta^2$ vs $C_{2022} \ge 5d\beta^3/\alpha$
            \item Upper Bound: $C_{2026} = 2d\beta^2$ vs $C_{2022} \ge 5d\beta^4/\alpha^2$
        \end{itemize}
    \end{itemize}
\end{enumerate}

\subsection{The Universal Topological Classifier}

Given any observation window $\Delta \in (0, \Delta_{max})$, we define a family of topological classifiers based on the scaled difference between the interaction and baseline empirical estimators. Specifically, the estimated synaptic nature of the link $j \to i$ is determined by the following Statistical Classifier:

\begin{equation} \label{test_statistico_adaptive}
\hat{S}^{j \to i}(m_0,m_1; \Delta) :=
\begin{cases}
1 & \text{if } \hat{\mathcal{G}}^{j \to i}_{m_0, m_1}(\Delta) > \frac{1}{2}, \quad (\text{Excitatory}) \\
-1 & \text{if } \hat{\mathcal{G}}^{j \to i}_{m_0, m_1}(\Delta) < -\frac{1}{2}, \quad (\text{Inhibitory}) \\
0 & \text{if } |\hat{\mathcal{G}}^{j \to i}_{m_0, m_1}(\Delta)| \le \frac{1}{2}, \quad (\text{Null})
\end{cases}
\end{equation}

Crucially, Theorem \ref{TEO-BOUND} guarantees that for any stability condition $\Delta < \Delta_{max}$, the theoretical expected value of $\hat{\mathcal{G}}$ lies within the correct decision region with an absolute safety margin of at least $\eta(\Delta) / \delta\Delta$. This property ensures that the universal classification thresholds $\zeta_{\pm} = \pm 1/2$ are valid across the entire stability range. Consequently, the identification error is exclusively driven by statistical finite-sample fluctuations, which we will rigorously analyze in the next section.

While the thresholds hold for any valid window, the specific choice of the parameter $\Delta$ reflects a fundamental trade-off in the estimator's design. Increasing $\Delta$ is generally desirable from a sampling perspective, as it allows for a larger number of events to be captured within each window, thus improving the convergence rate. However, as shown in Theorem \ref{TEO-BOUND}, the error terms grow linearly with $\Delta$, causing the theoretical safety margins to shrink. 

The value adopted in our framework, $\Delta_A = \frac{\delta}{9d\beta^2}$, represents a compromise: it ensures that the window is wide enough for practical data collection while maintaining a sufficiently large and balanced gap between the theoretical regions. By evaluating the bounds at this specific resolution (Corollary \ref{cor:identification}), the expected values become scale-free and strictly separated:
\begin{equation}\label{ben}
\mathcal{G}^{j \to i}(\Delta_A)
 \in 
\begin{cases} 
\left( \frac{2}{3} + \frac{1}{6d}, + \infty \right) \subset \left( \frac{1}{2} , +\infty \right) & \text{if } j \to i \text{ is excitatory,} \\
\left[ -\frac{1}{3} - \frac{1}{18d}, \frac{1}{3} + \frac{1}{18d} \right] \subset \left[-\frac{1}{2}, \frac{1}{2}\right] & \text{if } j \text{ is not pre-synaptic to } i, \\
\left( -\infty , -\frac{2}{3} - \frac{1}{6d} \right] \subset \left( -\infty , -\frac{1}{2} \right) & \text{if } j \to i \text{ is inhibitory.}
\end{cases}
\end{equation} 

By combining the ergodic limits \eqref{ergodico} and the bounds in \eqref{ben}, we conclude that the statistical classifier $\hat{S}_{ij}(m_0,m_1; \Delta_A)$ is asymptotically correct, ensuring that the probability of a classification error vanishes as $m_0, m_1 \to \infty$. As illustrated in Figure \ref{fig:decision_axis}, for $\Delta = \Delta_A$, the decision space is cleanly partitioned by the thresholds $\zeta_{\pm} = \pm 1/2$, leaving robust safety margins against fluctuations.

\begin{figure}[htbp]
\centering
\begin{tikzpicture}[scale=0.8, transform shape, x=6cm, y=0.7cm, >=Stealth]
    \fill[red!20, opacity=0.4] (-1.7, -0.8) rectangle (-0.5, 0.8);
    \fill[gray!10, opacity=0.4] (-0.5, -0.8) rectangle (0.5, 0.8);
    \fill[green!20, opacity=0.4] (0.5, -0.8) rectangle (1.7, 0.8);

    \draw[latex-latex, thick] (-1.6,0) -- (1.6,0) node[right] {${\mathcal{G}}^{j \to i}(\Delta_A)$};

    \draw[ultra thick, dashed, gray] (-0.5, -1) -- (-0.5, 1) node[above, black] {$-1/2$};
    \draw[ultra thick, dashed, gray] (0.5, -1) -- (0.5, 1) node[above, black] {$1/2$};

    \node at (-1.1, 1.5) {\small \textbf{Inhibitory Region}};
    \node at (0, 1.5) {\small \textbf{Null Region}};
    \node at (1.1, 1.5) {\small \textbf{Excitatory Region}};

    \filldraw[black] (0.35,0) circle (2pt) node[below=2pt] {\scriptsize $\frac{1}{3} + \frac{1}{18d}$};
    \filldraw[black] (-0.35,0) circle (2pt) node[below=2pt] {\scriptsize $-\frac{1}{3} - \frac{1}{18d}$};

    \filldraw[black] (0.65,0) circle (2pt) node[below=2pt] {\scriptsize $\ge \frac{2}{3} + \frac{1}{6d}$};
    \filldraw[black] (-0.65,0) circle (2pt) node[below=2pt] {\scriptsize $\le -\frac{2}{3} - \frac{1}{6d}$};

    \draw[<->, blue, thick] (0.36, 0.4) -- (0.49, 0.4) node[midway, above] {\tiny gap};
    \draw[<->, blue, thick] (0.51, 0.4) -- (0.64, 0.4) node[midway, above] {\tiny gap};
    
    \draw[<->, blue, thick] (-0.36, 0.4) -- (-0.49, 0.4) node[midway, above] {\tiny gap};
    \draw[<->, blue, thick] (-0.51, 0.4) -- (-0.64, 0.4) node[midway, above] {\tiny gap};

    \draw[->, semithick] (-0.6, -1.3) -- (-1.5, -1.3) node[midway, below] {\tiny inhibitory semi-ray};
    \draw[->, semithick] (0.6, -1.3) -- (1.5, -1.3) node[midway, below] {\tiny excitatory semi-ray};
\end{tikzpicture}
\caption{Classification map evaluated at $\Delta_A$. The fixed universal thresholds $\zeta_{\pm} = \pm 1/2$ naturally bisect the scale-free bounds, providing an equidistant safety margin (blue gaps) that ensures robust classification against finite-sample variations.}
\label{fig:decision_axis}
\end{figure}


\section{Consistency of the Statistical Classifier} \label{Sec:6}

In this section, we prove the asymptotic correctness of the statistical classifier $\hat{S}^{j \to i}(m_0,m_1; \Delta)$ introduced in \eqref{test_statistico_adaptive}. Specifically, we provide exponential bounds on the probability of misclassification by analyzing the concentration properties of the estimator $\hat{\mathcal{G}}^{j \to i}_{m_0, m_1}(\Delta)$ defined in \eqref{def_hat_G_cal}.

As illustrated in Figure \ref{fig:decision_axis}, the classification remains correct as long as the fluctuations of the empirical estimator do not exceed the safety gaps between the theoretical bounds \eqref{ben} and the decision thresholds $\pm 1/2$.

For analytical convenience, we study the concentration properties of the two additive terms in \eqref{def_hat_G_cal} separately. Our proof strategy is based on \textit{Stochastic Domination} and \textit{Large Deviation Theory} (LDT), partially following the approach in \cite[Theorem 2]{DeSantis2022}. However, we extend this analysis by considering a fixed number of events in the numerator of the first term, which shifts the analysis from a binomial distribution to a waiting-time distribution. We adopt this approach to maintain a uniform statistical sensitivity when comparing simulations at different resolutions $\Delta$.

\subsection{Concentration of the Interaction Term}

To control the probability that the empirical estimator falls into the wrong decision region, we first analyze the interaction term:
\begin{equation}
	\hat{\Phi}_{m_1}(\Delta) = \frac{\sum_{k=1}^{m_1} \mathbf{1}_{\mathcal{D}^{j \to i}_k(\Delta)}}{\sum_{k=1}^{m_1} \mathbf{1}_{\mathcal{C}^{j \to i}_k(\Delta)}} .
\end{equation}

To formalize the estimation with a fixed number of observed successes $n$, we define the stopping time $M_1(n)$ as:
\begin{equation} \label{def_M1}
	M_1(n) := \inf \left\{ m \in \mathbb{N} : \sum_{k=1}^m \mathbf{1}_{\mathcal{D}^{j \to i}_k(\Delta)} = n \right\}.
\end{equation}

By defining the estimator through $M_1(n)$, the interaction term can be expressed as the inverse of the sample mean of the waiting times:
\begin{equation} \label{interazione_xi}
	\hat{\Phi}_{M_1(n)}(\Delta) = \frac{n}{\sum_{k=1}^{M_1(n)} \mathbf{1}_{\mathcal{C}^{j \to i}_k(\Delta)}} = \frac{n}{\sum_{k=1}^n \xi_k} = (\bar{\xi}_n)^{-1}
\end{equation}
where $\xi_k$ represents the number of steps in $\mathcal{C}$ required to observe successive events in $\mathcal{D}$.

To handle the dependencies within the sequence $\{\xi_k\}$, we employ a stochastic domination argument. Recalling the analysis in Section \ref{Sec:5}, the conditional probability of observing a success in $\mathcal{D}^{j \to i}_k$ satisfies the following uniform bounds for a sufficiently small $\Delta > 0$:
\begin{equation} \label{bound_p_delta_recall}
	a\Delta - b\Delta^2 \le \mP(\mathcal{D}^{j \to i}_k(\Delta) | \mathcal{C}^{j \to i}_k(\Delta), \mathcal{F}_{\sigma^{j \to i}_k}) \le a\Delta + B\Delta^2,
\end{equation}
where the coefficients $a, b, B$ depend on the infinitesimal generator and the synaptic weights, as specified in Corollary \ref{cor:explicit_bounds}. By setting $p_{min} = a\Delta - b\Delta^2$ and $p_{max} = a\Delta + B\Delta^2$, each waiting time $\xi_k$ is stochastically sandwiched between two i.i.d. Geometric random variables:
\begin{equation}
	X_k^{min} \le_{st} \xi_k \le_{st} X_k^{max}
\end{equation}
where $X_k^{min} \sim \text{Geom}(p_{max})$ and $X_k^{max} \sim \text{Geom}(p_{min})$. Consequently, the partial sums satisfy:
\begin{equation} \label{sandwich_sums}
	\sum_{k=1}^{n} X_k^{min} \le_{st} \sum_{k=1}^{n} \xi_k \le_{st} \sum_{k=1}^{n} X_k^{max}.
\end{equation}

This construction allows us to transport the concentration properties of the i.i.d. geometric sequences to our estimator. In this framework, as $\Delta$ decreases, the bounding waiting times scale as $\Delta^{-1}$, effectively balancing the multiplicative factor $(\Delta \delta)^{-1}$ present in the definition of $\hat{\mathcal{G}}^{j \to i}_{m_0, m_1}(\Delta)$. This property facilitates the comparison of different simulation regimes, as the misclassification probability is governed by the sample depth $n$ rather than the specific value of $\Delta$.

\subsection{Concentration and Relative Error}

To evaluate the precision of the estimator $\hat{\Phi}_{M_1(n)}(\Delta)$, we define a relative error threshold $\gamma > 0$. We are interested in the probability that the empirical estimator deviates from the reference value $a\Delta$ by more than a factor $\gamma$.

By virtue of the stochastic domination established in \eqref{sandwich_sums}, the tail probabilities of the non-stationary sequence $\{\xi_k\}$ are controlled by the tail probabilities of the i.i.d. sequences $\{X_k^{min}\}$ and $\{X_k^{max}\}$. Let $\mu_{min} = (a\Delta + b\Delta^2)^{-1}$ and $\mu_{max} = (a\Delta - b\Delta^2)^{-1}$ be the exact expected waiting times for the dominating processes. The probability that the interaction term underestimates the discrete theoretical lower bound (e.g., in the excitatory case) by a relative error $\gamma \in (0,1)$ can be bounded as follows:
\begin{equation} \label{bound_gamma_dominato}
	\mP\left( \hat{\Phi}_{M_1(n)}(\Delta) < (a\Delta - b\Delta^2)(1 - \gamma) \right) \le \mP\left( \frac{1}{n}\sum_{k=1}^n X_k^{max} > \frac{\mu_{max}}{1-\gamma} \right).
\end{equation}

Similarly, the probability of overestimating the discrete theoretical upper bound is controlled by the lower tail of the $X_k^{min}$ process:
\begin{equation}\label{bound_gamma_dominato2}
	\mP\left( \hat{\Phi}_{M_1(n)}(\Delta) > (a\Delta + b\Delta^2)(1 + \gamma) \right) \le \mP\left( \frac{1}{n}\sum_{k=1}^n X_k^{min} < \frac{\mu_{min}}{1+\gamma} \right).
\end{equation}

Since our estimator is defined as the inverse of the sample mean, bounding the tail events of $\hat{\Phi}_{M_1(n)}(\Delta)$ mathematically translates to bounding the tail events of the sample sums. Therefore, by applying Cramér's Theorem for the sums of i.i.d. Geometric random variables, we can establish explicit exponential upper bounds for the probabilities in \eqref{bound_gamma_dominato} and \eqref{bound_gamma_dominato2}.

Recall that for a sequence of i.i.d. random variables $X_k \sim \text{Geom}(p)$, the large deviation rate function $I(x, p)$ is obtained via the Legendre transform of its cumulant generating function, yielding:
\begin{equation} \label{rate_function_geom}
	I(x, p) = x \ln\left(\frac{1}{px}\right) + (x-1) \ln\left(\frac{x-1}{x(1-p)}\right) \quad \text{for } x > 1.
\end{equation}
This rate function satisfies $I(1/p, p) = 0$ and is strictly positive for $x \neq 1/p$.

To obtain explicit exponential bounds, we analyze the rate function at the specific thresholds defined by the relative error $\gamma$. Let $p(\Delta) = a\Delta \pm b\Delta^2$ represent the success probabilities. We consider the thresholds $x_+ = [p_{max}(\Delta)(1-\gamma)]^{-1}$ for the upper tail and $x_- = [p_{min}(\Delta)(1+\gamma)]^{-1}$ for the lower tail. 

We define the composite rate function $H(p, \gamma)$ by evaluating $I(x, p)$ at these points:
\begin{equation}
	H(p, \gamma) := I\left(\frac{1}{p(1-\gamma)}, p\right) = \ln(1-\gamma) + \left( \frac{1}{p(1-\gamma)} - 1 \right) \ln\left( \frac{1-p(1-\gamma)}{1-p} \right).
\end{equation}

In the limit $\Delta \to 0$ (and thus $p \to 0$), the rate function converges to the zeroth-order term, which corresponds to the rate function of a continuous Exponential distribution:
\begin{equation}
	H(0, \gamma) = \lim_{p \to 0} H(p, \gamma) = \ln(1-\gamma) + \frac{\gamma}{1-\gamma}.
\end{equation}

To rigorously justify the adoption of $H(0, \gamma)$ as a conservative bound, we must prove that $H(p, \gamma)$ is strictly bounded from below by $H(0, \gamma)$. We achieve this by analyzing the Taylor expansion of $H(p, \gamma)$ with respect to $p$, centered at $p=0$, with the Lagrange remainder:
\begin{equation}
    H(p, \gamma) = H(0, \gamma) + p \left. \frac{\partial H(\xi, \gamma)}{\partial p} \right|_{p = \xi},
\end{equation}
for some $\xi \in (0, p)$. To show that $H(p, \gamma) > H(0, \gamma)$ for any valid probability $p > 0$, it is sufficient to prove that the first derivative is strictly positive for all $p \in (0,1)$.

Taking the partial derivative of $H(p, \gamma)$ with respect to $p$ yields:
\begin{align}
    \frac{\partial H(p, \gamma)}{\partial p} &= -\frac{1}{p^2(1-\gamma)} \ln\left( \frac{1-p(1-\gamma)}{1-p} \right) + \left( \frac{1}{p(1-\gamma)} - 1 \right) \frac{\gamma}{(1-p)(1-p(1-\gamma))} \nonumber \\
    &= \frac{1}{p^2(1-\gamma)} \left[ \frac{p\gamma}{1-p} - \ln\left( 1 + \frac{p\gamma}{1-p} \right) \right].
\end{align}
Let us define the auxiliary variable $z := \frac{p\gamma}{1-p}$. The derivative can be compactly rewritten as:
\begin{equation}
    \frac{\partial H(p, \gamma)}{\partial p} = \frac{1}{p^2(1-\gamma)} \Big[ z - \ln(1+z) \Big].
\end{equation}

For the upper tail ($\gamma > 0$), we have $z > 0$. For the lower tail, replacing $\gamma$ with $-\gamma$, we have $z \in (-1, 0)$ (since the threshold constraint implies $p(1+\gamma) < 1$). It is a well-known analytical property that the inequality $z > \ln(1+z)$ holds for all $z > -1, z \neq 0$. 

Since both $p^2$ and $(1-\gamma)$ are strictly positive, it follows immediately that $\frac{\partial H(p, \gamma)}{\partial p} > 0$ for all valid $p$ and $\gamma \neq 0$. Consequently, the rate function is strictly monotonically increasing with $p$, implying that the Lagrange remainder is strictly positive. This provides the rigorous guarantee that:
\begin{equation}
    H(p, \gamma) > H(0, \gamma), \quad \forall p \in (0,1).
\end{equation}
Thus, adopting the continuous limit places our bound in the strict worst-case scenario.

 Applying Cramér's Theorem, the concentration inequalities for the dominating processes become:
\begin{equation} \label{cramer_final_upper}
	\mP\left( \frac{1}{n}\sum_{k=1}^n X_k^{max} > \frac{\mu_{max}(\Delta)}{1-\gamma} \right) \le \exp\left( -n \left[ \ln(1-\gamma) + \frac{\gamma}{1-\gamma} \right] \right),
\end{equation}
and similarly for the lower tail:
\begin{equation}\label{cramer_final_lower}
	\mP\left( \frac{1}{n}\sum_{k=1}^n X_k^{min} < \frac{\mu_{min}(\Delta)}{1+\gamma} \right) \le \exp\left( -n \left[ \ln(1+\gamma) - \frac{\gamma}{1+\gamma} \right] \right).
\end{equation}

These bounds demonstrate that the error concentration depends strictly on the sample size $n$, while the discretization step $\Delta$ only shifts the mean values $\mu_{max}, \mu_{min}$.

\begin{figure}[htbp]
    \centering
    \begin{tikzpicture}
        \begin{axis}[
            title={Asymptotic Rate Function $H(0, \gamma)$},
            xlabel={Relative error $\gamma$},
            ylabel={$H(0, \gamma)$},
            xmin=-1, xmax=1,
            ymin=0, ymax=4,
            axis lines=left,
            domain=-0.99:0.8, 
            samples=200,
            thick,
            grid=both,
            grid style={line width=.1pt, draw=gray!20},
            major grid style={line width=.2pt,draw=gray!50},
            legend pos=north west
        ]
        
        \addplot [blue, very thick] {ln(1-x) + x/(1-x)};
        \addlegendentry{$\ln(1-\gamma) + \frac{\gamma}{1-\gamma}$}
        
        \draw[dashed, red] (axis cs:0,0) -- (axis cs:0,4);
        
        \end{axis}
    \end{tikzpicture}
    \caption{Behavior of the zeroth-order rate function $H(0, \gamma)$ for the continuous limit. The pronounced asymmetry illustrates that upper-tail large deviations ($\gamma > 0$) decay much faster than lower-tail deviations ($\gamma < 0$).}
    \label{fig:rate_function}
\end{figure}

Let us isolate the baseline component of the estimator. To maintain perfect symmetry with the interaction term, we define the baseline estimation through a fixed number of observed baseline spikes, $n_0$. We introduce the baseline stopping time:
\begin{equation} \label{def_M0}
	M_0(n_0) := \inf \left\{ m \in \mathbb{N} : \sum_{k=1}^m \mathbf{1}_{\mathcal{B}^i_k(\Delta)} = n_0 \right\}.
\end{equation}

By defining the estimator through $M_0(n_0)$, the baseline term can be naturally expressed as the inverse of the sample mean of the waiting times:
\begin{equation} \label{eq:baseline_estimator}
	\hat{\Phi}^i_{M_0(n_0)}(\Delta) := \frac{n_0}{M_0(n_0)} = \frac{n_0}{\sum_{k=1}^{n_0} \zeta_k} = (\bar{\zeta}_{n_0})^{-1},
\end{equation}
where $\zeta_k$ represents the number of trials (reset events $\mathcal{A}^i$) required to observe successive baseline spikes in $\mathcal{B}^i$. 

As anticipated, this term relies on a mathematical structure identical to the interaction component. $M_0(n_0)$ is stochastically bounded by the sum of $n_0$ independent geometric variables, each corresponding to a success probability $p_{\mathcal{B}}(\Delta) = a_{\mathcal{B}}\Delta \pm b_{\mathcal{B}}\Delta^2$. 

Consequently, its large deviation bounds mirror exactly those derived previously. One simply needs to substitute the sample size $n_1$ with $n_0$ and use the appropriate asymptotic mean $\mu_{\mathcal{B}}$.

\medskip 

In our simulations, we deliberately set $n_0 \gg n_1$. This choice is driven by two fundamental reasons rooted in the physics of the network. First, the baseline activity is intrinsically a node-specific process with its own characteristic timescale. Obtaining $n_1$ interacting events for a specific pair $j \to i$ requires significantly longer observation times compared to observing $n_0$ baseline events for neuron $i$, since the baseline probability $p_{\mathcal{B}}(\Delta)$ is generally much larger than the directed interaction probability. 

Second, this asymmetry in timescales provides a substantial statistical and computational advantage. The highly concentrated baseline estimate $	\hat{\Phi}^i_{M_0(n_0)}(\Delta) $ for a target neuron $i$ can be reused to evaluate the effective connectivity from any arbitrary source neuron $j$, $j'$, and so forth. Since the number of possible directed interactions scales quadratically with the number of nodes, adopting a large $n_0$ ensures a highly accurate, shared baseline term without incurring prohibitive simulation costs.

\subsection{Large Deviation Bounds for Topology Identification}

The concentration inequalities derived for the individual estimators $\hat{\Phi}_{M_1 (n_1)}$ and $	\hat{\Phi}^i_{M_0(n_0)}$ can now be combined to evaluate the robustness of the topological classifier $\hat{S}_{ij}$. For the sake of notational clarity, let us define the continuous-limit rate functions governing the lower and upper tail deviations of the estimators, as derived in \eqref{cramer_final_upper} and \eqref{cramer_final_lower}:
\begin{align}
    \mathcal{I}_-(\gamma) &:= \ln(1-\gamma) + \frac{\gamma}{1-\gamma}, \\
    \mathcal{I}_+(\gamma) &:= \ln(1+\gamma) - \frac{\gamma}{1+\gamma}.
\end{align}
Notice that the asymmetry of the geometric sums inherently yields $\mathcal{I}_-(\gamma) > \mathcal{I}_+(\gamma)$ for any given relative error $\gamma \in (0,1)$, meaning that extreme underestimations decay exponentially faster than overestimations.

To quantify the misclassification probability, we analyze the event where the empirical estimator $\hat{\mathcal{G}}^{j \to i}$ crosses the decision boundaries $\pm 1/2$. Let $\Delta < \Delta_{max}$ be the observation window, and let $\eta(\Delta)$ be the absolute safety margin derived from the strict inequalities in Theorem \ref{TEO-BOUND}:
\begin{equation}
    \eta(\Delta) = \frac{1}{2}\delta\Delta - \left(3d - \frac{3}{2}\right)\beta^2\Delta^2.
\end{equation}

A classification error occurs if the sum of the absolute statistical fluctuations of the interaction and baseline terms exceeds the theoretical gap $\eta(\Delta)$. We distribute this margin using a splitting parameter $\theta \in (0,1)$. To apply the concentration inequalities, we must express these absolute margins as relative errors $\gamma$ by normalizing them with respect to the expected event probabilities.

Notice that the continuous-limit rate functions $\mathcal{I}_-(\gamma)$ and $\mathcal{I}_+(\gamma)$ are strictly monotonically increasing with respect to $\gamma$. Consequently, the upper bounds on the misclassification probability are maximized (i.e., the theoretical worst-case scenario is reached) when the relative error $\gamma$ is minimized. Since $\gamma$ is inversely proportional to the true event probability, the minimum relative error occurs at the maximum possible firing intensity.

To ensure our bounds hold uniformly across the entire network, regardless of the unknown specific resting intensity $\phi_i(0)$, we replace the local intensity with the global upper bound $\beta$. This rigorous worst-case approximation yields the uniform relative fluctuations:
\begin{equation} \label{eq:gamma_defs}
    \gamma_1 := \theta \frac{\eta(\Delta)}{\beta\Delta}, \quad \gamma_0 := (1-\theta) \frac{\eta(\Delta)}{\beta\Delta}.
\end{equation}

\begin{remark} \normalfont In strictly physical terms, since the total instantaneous intensity cannot exceed the global boundary $\beta$, an excitatory connection implies $\phi_i(0) + \delta \le \beta$, which tightens the baseline upper bound to $\phi_i(0) \le \beta - \delta$. While using $\beta - \delta$ in the denominator of \eqref{eq:gamma_defs} would theoretically yield a slightly tighter bound for the relative fluctuations, we deliberately adopt the global parameter $\beta$ to maintain a symmetric and universally valid expression across all topological cases (excitatory, inhibitory, and null).
\end{remark}

\begin{theorem}[Classification Error Bounds] \label{thm:LD_error}
Under the stability condition $\Delta < \Delta_{max}$, let $\theta \in (0,1)$ be a splitting parameter and define the relative fluctuation margins as in \eqref{eq:gamma_defs}.
For any target sample sizes $n_1$ and $n_0$, the probability of misclassifying the synaptic connection $j \to i$ is bounded as follows:
\begin{enumerate}
    \item \textbf{Excitatory case ($j \in \mathcal{V}_+^i$):} The probability of a misclassification is bounded by:
    \begin{equation}
        \mP\left(\hat{S}^{j \to i}(M_0(n_0), M_1(n_1); \Delta) \neq 1\right) \le \exp\left( -n_1 \mathcal{I}_-( \gamma_1 ) \right) + \exp\left( -n_0 \mathcal{I}_+( \gamma_0 ) \right).
    \end{equation}
    
    \item \textbf{Inhibitory case ($j \in \mathcal{V}_-^i$):} The probability of a misclassification is bounded by:
    \begin{equation}
        \mP\left(\hat{S}^{j \to i}(M_0(n_0), M_1(n_1); \Delta) \neq -1\right) \le \exp\left( -n_1 \mathcal{I}_+( \gamma_1 ) \right) + \exp\left( -n_0 \mathcal{I}_-( \gamma_0 ) \right).
    \end{equation}
    
    \item \textbf{Null case ($j \notin \mathcal{V}^i$):} The probability of a false positive identification is bounded by:
    \begin{equation}
        \mP\left(\hat{S}^{j \to i}(M_0(n_0), M_1(n_1); \Delta) \neq 0\right) \le 2 \left[ \exp\left( -n_1 \mathcal{I}_+( \gamma_1 ) \right) + \exp\left( -n_0 \mathcal{I}_+( \gamma_0 ) \right) \right].
    \end{equation}
\end{enumerate}
\end{theorem}
\begin{proof}
We detail the proof for the excitatory case ($j \in \mathcal{V}_+^i$); the inhibitory and null cases follow by symmetric arguments. 

By the definition of the classifier in \eqref{test_statistico_adaptive}, a misclassification for an excitatory synapse occurs if $\hat{S}^{j \to i} \neq 1$, which strictly corresponds to the event $\hat{\mathcal{G}}^{j \to i}_{M_0, M_1}(\Delta) \le 1/2$. 
Recall that the estimator $\hat{\mathcal{G}}$ is constructed as the difference of the interaction and baseline empirical means. By construction of the safety margin $\eta(\Delta)$ and the splitting parameter $\theta$, for the estimated difference to fall below the threshold $1/2$, at least one of two rare events must occur: either the interaction estimator fluctuates below its theoretical expected value by a relative margin greater than $\gamma_1$, or the baseline estimator fluctuates above its expected value by a relative margin greater than $\gamma_0$. 

Using the union bound, the probability of this union of events is bounded by the sum of their individual probabilities, regardless of any temporal correlations between the two observation processes. Invoking the large deviation bounds derived previously for the individual empirical means, we obtain:
\begin{equation*}
    \mP\left(\hat{S}^{j \to i} \neq 1\right) \le \exp\left( -n_1 \mathcal{I}_-( \gamma_1 ) \right) + \exp\left( -n_0 \mathcal{I}_+( \gamma_0 ) \right).
\end{equation*}

The inhibitory case follows identically by reversing the signs. For the null case ($j \notin \mathcal{V}^i$), a false positive corresponds to $|\hat{\mathcal{G}}^{j \to i}_{M_0, M_1}(\Delta)| \ge 1/2$. Because the decision region is bounded on both sides, the estimator can exceed the threshold in either the positive or negative direction. Applying the union bound across both directions for both estimators yields the symmetric factor of $2$ in the final bound.
\end{proof}

\begin{remark}[Uniformity of the Bounds]  \label{rem-uniforme}\normalfont
The error bounds derived in 
Theorem \ref{thm:LD_error} depend on the discretization step $\Delta$ through the relative fluctuation margins $\gamma_1$ and $\gamma_0$. To better highlight the structural properties of these margins, let us introduce the dimensionless parameter $\tau := \delta/\beta$ (following the notation introduced in \cite{DeSantis2022}). Physically, $\tau$ acts as the intrinsic signal-to-noise ratio of the network: it is strictly invariant under any global time-rescaling of the transition rates (which merely accelerates or slows down the temporal dynamics without altering the statistical distinguishability of the events), thus capturing the true physical strength of a synaptic interaction.

By substituting the definition of the safety gap $\eta(\Delta)$ into \eqref{eq:gamma_defs}, we can make the dependence on $\Delta$ explicit:
$$
    \frac{\eta(\Delta)}{\beta\Delta} = \frac{1}{\beta} \left[ \frac{\delta}{2} - \left(3d - \frac{3}{2}\right)\beta^2\Delta \right] = \frac{\tau}{2} - \left(3d - \frac{3}{2}\right)\beta\Delta.
$$

Since the coefficient of $\Delta$ is strictly negative, the relative margins $\gamma_1(\Delta)$ and $\gamma_0(\Delta)$ are monotonically decreasing functions of the window size. 
This monotonicity allows us to establish uniform, numerical lower bounds valid for any window $\Delta \le \Delta_A$. By evaluating the margins at $\Delta_A := \frac{\delta}{9d\beta^2} = \frac{\tau}{9d\beta}$, we obtain the worst-case dimensionless gap:
$$
    \frac{\eta(\Delta_A)}{\beta\Delta_A} = \frac{\tau}{2} - \left(3d - \frac{3}{2}\right)\beta \left(\frac{\tau}{9d\beta}\right) = \tau \left[ \frac{1}{2} - \frac{3d - 1.5}{9d} \right] = \tau \left[ \frac{1}{2} - \frac{1}{3} + \frac{1}{6d} \right] = \frac{\tau}{6} \left( 1 + \frac{1}{d} \right).
$$

Therefore, for any valid choice of $\Delta \le \Delta_A$, the relative fluctuations are strictly bounded from below by network-dependent, dimensionless constants:
$$
    \gamma_1 \ge \theta \frac{\tau}{6}\left( 1 + \frac{1}{d} \right), \quad \gamma_0 \ge (1-\theta) \frac{\tau}{6}\left( 1 + \frac{1}{d} \right).
$$

Since the rate functions $\mathcal{I}_-$ and $\mathcal{I}_+$ are strictly monotonically increasing, evaluating them at these uniform lower bounds yields a universal exponential decay rate. This proves that the statistical consistency of the classifier is universally guaranteed and structurally decoupled from the specific choice of $\Delta$, provided the stability condition $\Delta \le \Delta_A$ is satisfied.
\end{remark}

\subsection{Sample Complexity and Physical Observation Time}

The large deviation bounds derived in Theorem \ref{thm:LD_error} allow us to establish the information-theoretic complexity of the network reconstruction problem. By inverting the exponential bounds, we can determine the minimum number of events required to achieve a target confidence level and, consequently, the expected physical time needed for the experiment.

To simplify the analysis and provide a conservative estimate, we fix the relative fluctuation margins $\gamma_1$ and $\gamma_0$ based on the reference window $\Delta_A$, even when evaluating the performance for smaller observation windows $\Delta < \Delta_A$.

\begin{corollary}[Expected Observation Time] \label{cor:time_complexity}
Given a target error $\epsilon > 0$ and $\Delta \leq \Delta_A $, let $\gamma_1$ be the relative margin at $\Delta_A$. The expected physical time $\mathbb{E}[T_{obs}]$ required to identify a connection with confidence $1-\epsilon$ is bounded by:
\begin{equation} \label{eq:T_obs_complexity}
    \mathbb{E}[T_{obs}(\epsilon, \Delta)] \le \frac{\ln(2/\epsilon)}{\alpha^3 \Delta^2 \mathcal{I}_{-}(\gamma_1)}.
\end{equation}
\end{corollary}

\begin{proof}
To ensure the error is bounded by $\epsilon$, the required number of independent events is $n \ge \ln(2/\epsilon) / \mathcal{I}_{-}(\gamma_1)$. In the high-resolution regime ($\Delta \ll 1/\alpha$), a single "success" requires the coincidence of arrivals within a pattern of total width $2\Delta$. Given the lower bound $\alpha$ on firing intensities, the effective rate of these patterns is $\lambda_{eff} \approx \alpha^3 \Delta^2$. 

The expected physical time to collect one such pattern is $\mathbb{E}[T_1] = 1/\lambda_{eff}$. Since $T_{obs} \gg 2\Delta$, the total time required to accumulate $n$ events is dominated by the inter-event times, $\mathbb{E}[T_{obs}] \approx n \cdot \mathbb{E}[T_1]$, leading to the upper bound in \eqref{eq:T_obs_complexity}.
\end{proof}

\begin{remark}[Sufficient Scaling Law for Reconstruction] \normalfont
As a direct consequence of the uniform bounds established above, we can make the dependence of the guaranteed observation time on the intrinsic dimensionless parameters of the network explicit. 

For small values of the relative margins, the large deviation rate function exhibits a locally quadratic lower bound, $\mathcal{I}_-(\gamma_1) = \mathcal{O}(\tau^2)$. Consequently, the sample size sufficient to bound the error scales as $n \propto 1/\tau^2$. By substituting this into the physical time bound of Corollary \ref{cor:time_complexity}, we obtain a preliminary upper bound for a generic, fixed window $\Delta$:
$$
    \mathbb{E}[T_{obs}(\Delta)] = \mathcal{O}\left( \frac{1}{\alpha^3 \Delta^2 \tau^2} \right).
$$

However, the observation window cannot be arbitrarily fixed. To rigorously guarantee statistical consistency and prevent systematic bias, the stability condition requires shrinking the window proportionally to the signal, $\Delta \le \Delta_A = \tau / (9d\beta)$. 
By substituting $\Delta_A$ into the temporal bound and recalling the dimensionless dynamic range parameter $s := \alpha/\beta \in (0,1]$ (as in \cite{DeSantis2022}), we extract the timescale $1/\beta$ and establish a structural upper bound on the expected observation time:
$$
    \mathbb{E}[T_{obs}(\Delta_A)] \le \mathcal{O}\left( \frac{1}{\alpha^3 \left(\frac{\tau}{9d\beta}\right)^2 \tau^2} \right) = \mathcal{O}\left( \frac{1}{\beta} \frac{d^2}{s^3 \tau^4} \right).
$$

This factorization elegantly separates the strictly physical timescale from the information-theoretic upper bound. The prefactor $1/\beta$ carries the physical dimensions of time, acting as the basal timescale of the network's background activity. The dimensionless block dictates a sufficient ``simulation time'' to achieve exponential concentration. 

Notably, this guaranteed complexity bound exhibits a severe \textit{quartic} dependence on the inverse signal-to-noise ratio ($1/\tau^4$). While this is fundamentally an upper bound and the exact physical time might scale less severely, this quartic form effectively captures a compound penalty intrinsic to our bounding technique for point-process reconstruction. Specifically, a weaker connection not only inflates the required sample size due to tighter statistical margins (the $1/\tau^2$ factor), but physically forces a proportional shrinking of the valid observation window ($\Delta_A \propto \tau$). This shrinking drastically reduces the expected occurrence rate of joint synaptic patterns, introducing a second quadratic toll. Thus, the bound fully quantifies a sufficient informational toll required to overcome background interference, scaling quadratically with the topological degree $d$ and with the inverse cube of the dynamic range ($1/s^3$).
\end{remark}

%

\subsection{Numerical Example: Practical Implementation}

To illustrate the implications of our theoretical bounds, we consider a generic high-activity neural network reconstruction task with the following reference parameters: $\alpha = 20$ Hz, $\beta = 30$ Hz, $\delta = 12$ Hz, and a maximum local degree $d = 2$. This yields an intrinsic signal-to-noise ratio $\tau = 0.4$.

We determine the necessary observation time to guarantee a target classification error $\epsilon = 5\%$ using the rigorous upper bound derivations. First, we compute the maximum stability window:
$$
    \Delta_A = \frac{\tau}{9d\beta} = \frac{0.4}{9 \cdot 2 \cdot 30} \approx 7.4 \times 10^{-4} \text{ s} \approx 0.74 \text{ ms}.
$$

By setting the threshold parameter $\theta = 0.9$ to heavily favor the alternative hypothesis, the uniform worst-case margin evaluated at $\Delta_A$ is:
$$
    \gamma_1 = \theta \frac{\tau}{6}\left( 1 + \frac{1}{d} \right) = 0.9 \cdot \frac{0.4}{6} \cdot 1.5 = 0.09.
$$

To avoid the underestimation of the decay rate caused by quadratic Taylor approximations, we evaluate the exact Cramér rate function derived in \eqref{cramer_final_lower} for the lower tail:
$$
    \mathcal{I}_-(0.09) = \ln(1 - 0.09) + \frac{0.09}{1 - 0.09} \approx 0.00459.
$$

The sufficient sample size to ensure confidence $1-\epsilon$ is therefore:
$$
    n_1 \ge \frac{\ln(2/0.05)}{0.00459} \approx 804 \text{ events}.
$$

Applying the physical time bound of Corollary \ref{cor:time_complexity}, the effective rate of joint patterns within the window $\Delta_A$ is $\lambda_{eff} = \alpha^3 \Delta_A^2 \approx 4.38 \times 10^{-3}$ Hz. The structural upper bound on the expected physical time is:
$$
    \mathbb{E}[T_{obs}] \le \frac{804}{4.38 \times 10^{-3}} \approx 183,560 \text{ s} \approx 51 \text{ hours}.
$$

\begin{remark}[Worst-case vs. Practical Performance] \normalfont
The derived bound of $\mathbb{E}[T_{obs}] \approx 51$ hours serves as a \textbf{rigorous sufficient condition} ensuring convergence under maximally adversarial conditions. In practical simulations, however, successful reconstruction is typically achieved in a fraction of this time due to a cascade of conservative analytical steps inherent to our bounding technique:
\begin{enumerate}
    \item \textbf{Accumulation of Worst-Case Bounds:} Our analytical framework rigorously stacks worst-case assumptions at multiple levels. From bounding the global topological interference with the maximum degree $d$, to the intrinsically conservative nature of large deviation inequalities. Furthermore, the maximum stability window $\Delta_A$ is an extreme analytical limit; in reality, the statistical gap does not vanish abruptly for $\Delta > \Delta_A$, allowing the safe empirical use of significantly larger windows.
    \item \textbf{Dynamical Forgiveness:} The theoretical bounds conservatively assume that background noise completely masks the synaptic jump unless the membrane potential is perfectly unperturbed. In practice, the linear gain induced by $\delta$ leaves a statistically detectable signature even when the baseline potential fluctuates above zero.
    \item \textbf{Macro-Micro Extrapolation:} Evaluating the empirical estimator $\hat{\mathcal{G}}(\Delta)$ across multiple window sizes and extrapolating for $\Delta \to 0$ isolates the true synaptic jump from the baseline noise, effectively acting as a powerful computational denoising filter.
\end{enumerate}
Therefore, the derived theoretical time scale represents the fundamental structural cost of absolute certainty, whereas practical simulations comfortably exploit the average-case physical dynamics of the network.
\end{remark}


\section{Numerical Validation of the Spike-Triggered Estimator} \label{Sec:7} 

In this section, we provide a numerical validation of the inference framework. The objective is to verify that the Spike-Triggered Estimator $\hat{\mathcal{G}}^{j \to i }_{M_0, M_1}(\Delta)$ converges correctly and that the decision thresholds $ \pm 1/2$ provide a robust separation of the connectivity types under the predicted sample complexity.

\subsection{Simulation Setup} 
We adopt a global activation function $\phi(u)$ for all neurons $i \in I$, defined as the following piecewise linear function:
\begin{equation} \label{function:phi}
\phi(u) =\begin{cases}
\alpha & \text{if } u \le u_{low}, \\ 
\beta  & \text{if } u \ge u_{high}, \\ 
\alpha + (u - u_{low})\frac{\beta - \alpha}{u_{high} - u_{low}} & \text{if } u_{low} < u < u_{high}.
\end{cases}
\end{equation}  
For these simulations, we set $\alpha = 1$, $\beta = 5$, $u_{low} = -2.0$, and $u_{high} = 2.0$. This configuration implies $\phi(0) = 3$ (the firing intensity at the reset potential). To validate the classification framework, we select synaptic weights $w^{j \to i} \in \{1, -1, 0\}$ for excitatory, inhibitory, and null connections, respectively.

The system is composed of $N$ nodes, where Node 0 is the target (post-synaptic) neuron, while the remaining nodes act as independent "driver" nodes (Poisson processes with intensity 3) that exert influence only on Node 0.

\subsection{Results and Discussion}
The numerical results, summarized in Table \ref{tab:combined_results}, confirm that the Spike-Triggered Estimator acts as a reliable topological classifier. In our simulations, we maintained a sparse connectivity with a local in-degree $d=2$ or $d=4$.

The results verify a key theoretical intuition: the systematic error (bias) of the estimator depends primarily on the local neighborhood $\mathcal{V}^i$ and the window $\Delta$, rather than the global size $N$. As long as $d$ remains small, the quadratic interference term $d\beta^2\Delta^2$ remains controlled, and the separation gap between different connection types remains sharp. 

The increase in the number of triggers to $n_1 = 300$ for the $N=10$ case was used to ensure robust statistical convergence. The estimates for null links remain consistently near zero ($|\hat{\mathcal{G}}| \approx 0.02$), proving that the estimator correctly ignores "indirect" connections.

\begin{table}[ht]
\centering
\caption{Classification results for the Spike-Triggered Estimator across different network sizes and window widths.}
\label{tab:combined_results}
\begin{tabular}{lcccccc}
\hline
\textbf{N} & \textbf{$n_1$} & \textbf{Link Type} & $\Delta$ & \textbf{$\hat{\mathcal{G}}$} & \textbf{Threshold} & \textbf{Outcome} \\ \hline
$4$  & 100 & Excitatory ($+1$) & 0.055 & 0.8493  & $\pm 0.5$ & \checkmark Correct \\
$4$  & 100 & Null ($0.0$)       & 0.055 & 0.0472  & $\pm 0.5$ & \checkmark Correct \\
$4$  & 100 & Excitatory ($+1$) & 0.009 & 0.9285  & $\pm 0.5$ & \checkmark Correct \\
$10$ & 300 & Excitatory ($+1$) & 0.009 & 1.1173  & $\pm 0.5$ & \checkmark Correct \\
$10$ & 300 & Inhibitory ($-1$) & 0.009 & -0.8316 & $\pm 0.5$ & \checkmark Correct \\
$10$ & 300 & Null ($0.0$)       & 0.009 & -0.0276 & $\pm 0.5$ & \checkmark Correct \\ \hline
\end{tabular}
\end{table}

While we verified that for $d \leq 4$ the estimator correctly classifies links using window sizes $\Delta \approx \frac{\delta}{\beta^2 d}$, adhering to the ultra-fine resolution $\Delta_A$ suggested by the conservative bounds of Corollary \ref{cor:identification} poses significant practical limitations. For instance, with $d=9$, $\Delta_A \approx 0.00049$ would necessitate computationally prohibitive simulation timescales. 

For this reason, the single-scale Spike-Triggered approach serves as the foundation for the Macro-Micro methods, which overcome these limitations by investigating the asymptotic behavior of $\hat{\mathcal{G}}^{j \to i} (\Delta)$ as $\Delta \to 0^+$, allowing for robust inference even with broader, more practical windows.


\section{The Macro-Micro Strategy}\label{Sec:8}

The proposed inference framework aims to reconstruct microscopic synaptic connections from observations at macroscopic time scales. This is achieved by evaluating the Spike-Triggered Estimator $\hat{\mathcal{G}}^{j \to i}_{M_0, M_1}(\Delta_k)$  over a multi-scale temporal grid 
$\{\Delta_k\}_{k=1}^5$ and employing an adaptive extrapolation-averaging logic to recover the synaptic signal.

\subsection{Multi-Scale Calibration} 
The temporal scales are chosen in a geometric progression $\Delta_k = (\sqrt{2})^{k-1} \Delta_1$. The initial window $\Delta_1$ is analytically determined to ensure the algorithm operates within a statistically significant yet linear regime:
\begin{equation}\label{Delta_1}
  \Delta_1 = \underbrace{\frac{1}{\beta d}}_{\text{Noise Scale}} \cdot \underbrace{\frac{\delta}{\beta}}_{\text{Sensitivity}} \cdot \underbrace{\frac{\beta - \alpha}{2\delta}}_{\text{Survival distance}} = \frac{\beta - \alpha}{2 d \beta^2}.
\end{equation}
By exploiting the knowledge of the intensity function $\phi(u)$, this calibration ensures that $\Delta_1$ is large enough to provide signal but small enough to avoid the saturation boundaries of the membrane potential.

\subsection{The Extrapolation Models}
To decouple the microscopic signal from macroscopic interference, we consider two complementary approaches for analyzing the evolution of $\mathcal{G}^{j \to i}(\Delta)$.

\paragraph{1. Local Quadratic Approximation}
We assume that for small $\Delta$, the gain function follows a second-order expansion:
\begin{equation} \label{ba}
    {\mathcal{G}}^{j \to i} (\Delta) \approx a \Delta^2 + b \Delta + c,
\end{equation}
where $c$ represents the direct synaptic effect, $b\Delta$ accounts for first-order interference from other presynaptic neurons, and $a\Delta^2$ aggregates global network noise and non-linearities.

\paragraph{2. Pyramid Extrapolation ($P$)}
To avoid the fragility and overfitting risks of direct parabolic fitting, we employ a recursive geometric construction designed for robust intercept recovery. This method, termed \textbf{Pyramid Extrapolation}, operates through two phases:

\begin{itemize}
    \item \textbf{Contraction via Barycentric Iteration:} Let $P^{(0)}_k = (\Delta_k, \hat{\mathcal{G}}^{j \to i}_{M_0, M_1}(\Delta_k))$ for $k=1, \dots, 5$ be the initial multi-scale dataset. At each iteration step $m \in \{1, 2, 3\}$, the algorithm computes the midpoints of adjacent pairs from the previous step:
    \begin{equation}
        P^{(m)}_k = \frac{P^{(m-1)}_k + P^{(m-1)}_{k+1}}{2}, \quad \text{for } k=1, \dots, 5-m.
    \end{equation}
    After three iterations, this hierarchical averaging condenses the five initial samples into exactly two robust ``meta-points'', denoted as $A = P^{(3)}_1$ and $B = P^{(3)}_2$. 
    
    \item \textbf{Linear Projection and Intercept Recovery:} The final estimate $P$ is defined as the y-intercept of the line passing through meta-points $A$ and $B$.
\end{itemize}

The hierarchical structure of the averages captures the \textit{implicit curvature} of the underlying gain function. This provides an intrinsic regularization of the signal, effectively acting as a discrete low-pass filter. It allows the model to track the fundamental trend required for synaptic classification while remaining resilient to high-frequency stochastic fluctuations (outliers) that plague measurements at extremely small scales.

\subsection{Hybrid Decision Logic: Pyramid vs. Mean} 
The core of our strategy lies in the adaptive selection between the Pyramid Extrapolation ($P$) and the Sample Mean ($M$) of the gains. Let $\mathcal{S} = \{+1, 0, -1\}$ be the set of theoretical synaptic states. The final index $I$ is chosen via a distance-minimization rule:
\begin{equation}
I = \arg\min_{X \in \{P, M\}} \left( \text{dist}(X, \mathcal{S}) \right).
\end{equation}

This hybrid approach balances stability and sensitivity based on the network's regime:
\begin{itemize}
    \item \textbf{The Mean ($M$)} is superior in \textbf{balanced noise regimes}, where stochastic fluctuations tend to cancel out. It provides a grounded, stable estimate, especially for null ($w=0$) links.
    \item \textbf{The Pyramid ($P$)} is essential in \textbf{unbalanced regimes} (e.g., predominantly excitatory drive). The background activity biases the gains; the pyramid effectively ``strips away'' this background slope, allowing the true micro-scale intercept to emerge.
\end{itemize}

\subsection{Classification} 
To prioritize the suppression of false positives, we adopt a conservative thresholding logic with an explicit safety buffer:
\begin{equation}      
\hat{\Sigma}_{ij} :=
\begin{cases} 
1 & \text{if } I > 5/8, \\
-1 & \text{if } I < -5/8, \\
0 & \text{if } |I| \le 5/8.
\end{cases}
\end{equation}
While a symmetric threshold at $\pm 1/2$ would be theoretically sufficient, the expansion of the null region to $\pm 5/8$ ensures that residual stochastic fluctuations in larger networks do not trigger misclassifications.


\section{Numerical Validation for Macro-Micro Extrapolation}\label{Sec:9}

\noindent
To validate the robustness of our framework under extreme conditions, we consider a target neuron ($i=0$) influenced by three presynaptic units. Neuron 1 maintains a constant intensity $ \frac{\alpha + \beta}{2}$, while Neurons 2 and 3 aggregate the activity of excitatory ($L_+$) and inhibitory ($L_-$) populations, respectively. 

This simplified architecture, illustrated in Fig. \ref{fig:min_arch}, provides a controlled environment to stress the framework under conditions of extremely low Signal-to-Noise Ratio (SNR). By setting $L_+, L_- \gg \beta$, we simulate a dominant stochastic background drive that would typically obscure synaptic signals in standard observation windows.

\begin{figure}[htbp] 
    \centering
    \begin{tikzpicture}[
        node distance=1.5cm,
        every node/.style={circle, draw, minimum size=0.6cm, font=\small},
        edge/.style={->, >={Stealth[length=2.5mm]}, thick}
    ]
        \node (0) {0};
        \node (1) [above left=of 0, label=left:{$\frac{\alpha+\beta}{2}$}] {1};
        \node (2) [above=of 0, label=above:{$L_+$}] {2};
        \node (3) [above right=of 0, label=right:{$L_-$}] {3};
    
        \draw[edge] (1) -- node[below left, draw=none] {$w$} (0);
        \draw[edge] (2) -- node[right, draw=none] {$+1$} (0);
        \draw[edge] (3) -- node[below right, draw=none] {$-1$} (0);
    \end{tikzpicture}
    \caption{Minimal network architecture for validation.}
    \label{fig:min_arch}
\end{figure}

\subsection*{Computational Advantages of the Minimal Network Architecture}

The selection of this topology offers significant advantages for event-driven simulations using Lewis' thinning algorithm:

\begin{enumerate}
    \item \textbf{Reduced State Update Overhead:} Since the intensity functions follow predictable trajectories between spikes, the calculation of $\lambda_{max}$ becomes computationally trivial, minimizing the overhead of variable tracking.
    \item \textbf{Efficiency in Event Attribution:} In this 4-neuron configuration, the search for the triggered unit among candidate spikes requires evaluating only four probabilities, ensuring that the computational cost per event remains extremely low.
    \item \textbf{Immediate Realization of the Reset State:} By assuming an initial state $u=0$ for the target neuron, we bypass the initial transient phases. Since presynaptic inputs are independent Poisson processes, the system is natively in statistical equilibrium, providing immediate access to post-reset dynamics.
\end{enumerate}

\begin{figure}[htbp]
    \centering
    \begin{subfigure}[b]{0.48\textwidth}
        \centering
        \begin{tikzpicture}
            \begin{axis}[
                width=1.0\textwidth, height=0.8\textwidth,
                xlabel={$\Delta$ (ms)}, ylabel={$\hat{\mathcal{G}}^{1 \to 0}(\Delta)$},
                xmin=-0.5, xmax=6, ymin=-1.5, ymax=2.0,
                grid=major, grid style={dashed, gray!30},
                title={\small Excitatory ($W=1.0$)},
                axis lines=left,
                x filter/.code={\pgfmathparse{#1*1000}\pgfmathresult},
                xtick={0, 1, 2, 3, 4, 5, 6},
                xticklabel style={/pgf/number format/fixed, font=\tiny},
                scaled x ticks=false
            ]
                \addplot[only marks, mark=*, blue, mark size=1.3pt] coordinates {
                    (0.001317, 1.1106) (0.001862, 1.3392) (0.002634, 1.4995) 
                    (0.003725, 1.6961) (0.005268, 1.8930)
                };
                \addplot[only marks, mark=o, red, thick, mark size=1.8pt] coordinates {
                    (0.002819, 1.5133) (0.003498, 1.6457)
                };
                \addplot[domain=0:0.006, green!60!black, dashed, thick] {194.93*x + 0.9638};
                \addplot[domain=-0.0005:0.006, red, dotted, thick] {0.625};
                \filldraw[black] (axis cs: 0, 0.9638) circle (1.8pt) node[anchor=south west, font=\tiny] {0.964};
            \end{axis}
        \end{tikzpicture}
    \end{subfigure}
    \hfill
    \begin{subfigure}[b]{0.48\textwidth}
        \centering
        \begin{tikzpicture}
            \begin{axis}[
                width=1.0\textwidth, height=0.8\textwidth,
                xlabel={$\Delta$ (ms)}, ylabel={},
                xmin=-0.5, xmax=6, ymin=-1.5, ymax=2.0,
                grid=major, grid style={dashed, gray!30},
                title={\small Null ($W=0.0$)},
                axis lines=left,
                x filter/.code={\pgfmathparse{#1*1000}\pgfmathresult},
                xtick={0, 1, 2, 3, 4, 5, 6},
                xticklabel style={/pgf/number format/fixed, font=\tiny},
                scaled x ticks=false
            ]
                \addplot[only marks, mark=*, blue, mark size=1.3pt] coordinates {
                    (0.001317, -0.0195) (0.001862, 0.0851) (0.002634, -0.0277) 
                    (0.003725, 0.1064) (0.005268, 0.0554)
                };
                \addplot[only marks, mark=o, red, thick, mark size=1.8pt] coordinates {
                    (0.002819, 0.0400) (0.003498, 0.0504)
                };
                \addplot[domain=0:0.006, green!60!black, dashed, thick] {15.30*x - 0.0031};
                \addplot[domain=-0.0005:0.006, red, dotted, thick] {0.625};
                \addplot[domain=-0.0005:0.006, red, dotted, thick] {-0.625};
                \filldraw[black] (axis cs: 0, -0.0031) circle (1.8pt) node[anchor=south west, font=\tiny] {-0.003};
            \end{axis}
        \end{tikzpicture}
    \end{subfigure}

    \vspace{0.5cm} 

    \begin{subfigure}[b]{0.48\textwidth}
        \centering
        \begin{tikzpicture}
            \begin{axis}[
                width=1.0\textwidth, height=0.8\textwidth,
                xlabel={$\Delta$ (ms)}, ylabel={$\hat{\mathcal{G}}^{1 \to 0}(\Delta)$},
                xmin=-0.5, xmax=6, ymin=-1.5, ymax=2.0,
                grid=major, grid style={dashed, gray!30},
                title={\small Inhibitory ($W=-1.0$)},
                axis lines=left,
                x filter/.code={\pgfmathparse{#1*1000}\pgfmathresult},
                xtick={0, 1, 2, 3, 4, 5, 6},
                xticklabel style={/pgf/number format/fixed, font=\tiny},
                scaled x ticks=false
            ]
                \addplot[only marks, mark=*, blue, mark size=1.3pt] coordinates {
                    (0.001317, -0.7941) (0.001862, -0.6580) (0.002634, -0.4950) 
                    (0.003725, -0.2595) (0.005268, 0.0019)
                };
                \addplot[only marks, mark=o, red, thick, mark size=1.8pt] coordinates {
                    (0.002819, -0.4597) (0.003498, -0.3188)
                };
                \addplot[domain=0:0.006, green!60!black, dashed, thick] {207.54*x - 1.0448};
                \addplot[domain=-0.0005:0.006, red, dotted, thick] {-0.625};
                \filldraw[black] (axis cs: 0, -1.0448) circle (1.8pt) node[anchor=north west, font=\tiny] {-1.045};
            \end{axis}
        \end{tikzpicture}
    \end{subfigure}

    \caption{Synaptic extrapolation results. (Top Left) Excitatory, (Top Right) Null, (Bottom) Inhibitory. Blue marks: raw gains; Open red circles: pyramid meta-points; Dashed line: extrapolation to $\Delta \to 0$. Red dotted lines indicate the classification thresholds.}
    \label{fig:trittico_final_thresholds}
\end{figure}

\subsection*{Experimental Framework and Results}
We conducted a systematic campaign of 810 independent simulations across 27 operational scenarios, defined by the interplay of $\alpha$, $\beta$, the load configurations $L_{\pm}$, and the weights $w$ of Neuron 1.

\begin{table}[h]
\centering
\caption{Simulation Parameters and Performance Results}
\label{tab:simulation_parameters}
\begin{tabular}{ll}
\toprule
\textbf{Parameter} & \textbf{Values / Conditions} \\
\midrule
Total Simulations & 810 \\
Scenarios & 27 ($3 \beta \text{ levels} \times 3 \text{ Noise types} \times 3 w \text{ weights}$) \\
Noise Types & Excitatory-only, Balanced, Inhibitory-only \\
Weights ($w$) & $\{+1, 0, -1\}$ \\
Success Rate & \textbf{100\%} \\
\bottomrule
\end{tabular}
\end{table}

The empirical results highlight the practical effectiveness of the \textbf{Pyramid Extrapolation}. In all 810 trials, the model consistently achieved correct classification. Even in the most extreme regime ($\beta=7$, aggregate noise of $400$ Hz), the iterative averaging effectively isolates the synaptic sign.

As illustrated in Figure \ref{fig:trittico_final_thresholds}, the stability of the recovered intercepts ($0.964$ for excitatory, $-1.045$ for inhibitory) suggests that this heuristic extrapolation effectively absorbs the bias induced by finite observation windows. These findings validate the method as a robust and efficient tool for reconstructing connectivity using only the fundamental parameters of the system: the intensity bounds $\alpha, \beta$, the synaptic impulse $\delta$, and the local sparsity $d$.

\section{Numerical Validation in Interactive Networks} \label{Sec:10}

To evaluate the effectiveness of the proposed Macro-Micro strategy, we transition from the minimal architecture to a fully interactive network of $N=20$ neurons. In this setting, the background noise is no longer an external Poisson process but an emergent property of the network dynamics, governed by the endogenous activity of the intensity functions $\phi(u)$ defined in \eqref{function:phi}.

\subsection{Simulation Setup and Statistical Protocol}
The network was configured with $\alpha = 1.0$, $\beta = 5.0$, and $N=20$ ($d=19$). Synaptic weights were assigned stochastically with probabilities $P(w=+1)=0.25$, $P(w=0)=0.5$, and $P(w=-1)=0.25$, obtaining $\delta = 1 $ in this first simulation. In a second simulation, we will take variable weights with a resulting $\delta$ equal to  0.5. 

\textbf{Fixed-Event Sampling Strategy:} To ensure uniform statistical power across all estimated links and diverse network scenarios, the simulation is not bound by a fixed time window. Instead, we implement a protocol where data accumulation for each pair $(j,i)$ proceeds until a target count of events is reached for both the conditional and the baseline observations. Specifically, we monitor the numerators of our differential estimator:
\begin{equation}
\sum_{k=1}^{m_1} \mathbf{1}_{\{ \mathcal{D}^{j \to i}_k(\Delta) \}} = n_1 \quad \text{and} \quad \sum_{k=1}^{m_0} \mathbf{1}_{\{ \mathcal{B}^{j \to i}_k(\Delta) \}} = n_0
\end{equation}
where we set $n_1 = 2000$ and $n_0 = 40000$. While $n_1$ is chosen to provide a robust yet computationally efficient sample for the conditional firing probability, we utilize a much larger $n_0$ for the baseline. This asymmetry is strategically beneficial: since baseline events are intrinsically more abundant and easier to sample, the high value of $n_0$ effectively suppresses the statistical noise of the subtraction term.  By maintaining these constant target counts, the estimator natively satisfies the theoretical bounds derived in Section \ref{Sec:6}. Consequently, the statistical variance remains rigorously controlled purely by $n_1$ and $n_0$, ensuring consistent confidence intervals across the entire spectrum of synaptic weights, strictly independent of the underlying firing rates and the local network regime.

For these experiments, the initial sampling window $\Delta_1$ was set following \eqref{Delta_1}:
\begin{equation}
  \Delta_1 = 2 \cdot \left( \frac{\beta - \alpha}{2 d \beta^2} \right) \approx 0.0042.
\end{equation}
A scaling factor of 2 was applied to slightly broaden the observation window, accelerating the accumulation of the required $n_1$ events.

\subsection{Performance and Hybrid Logic Selection}
We conducted 30 independent trials (10 per synaptic class). The objective was to recover the class of a specific synapse using the Hybrid Decision Logic (Pyramid vs. Mean). The results, summarized in Table \ref{tab:results_n20}, demonstrate 100\% accuracy.

\begin{table}[h]
\centering
\caption{Classification Performance and Index Recovery ($N=20$).}
\label{tab:results_n20}
\begin{tabular}{lcccc}
\toprule
\textbf{Ground Truth ($w$)} & \textbf{Mean Recovered Index} & \textbf{Std. Dev.} & \textbf{Accuracy} \\
\midrule
$+1.0$ (Excitatory) & $0.9893$ & $0.048$ & 100\% \\
$0.0$ (Null)        & $0.0049$ & $0.045$ & 100\% \\
$-1.0$ (Inhibitory) & $-0.9598$ & $0.031$ & 100\% \\
\bottomrule
\end{tabular}
\end{table}

\subsection{Discussion of Trial Dynamics}
A key finding is the adaptive behavior of the classifier based on the local network regime. For \textbf{Excitatory Connections} ($w=+1$), the system frequently selected the Sample Mean (60\% of trials). This occurs because the interaction maintains a stable positive pressure; as the window $\Delta$ expands, the signal remains coherent, allowing the mean to achieve optimal variance reduction.

Conversely, the \textbf{Pyramid Extrapolation} becomes essential in "disruptive" or slope-heavy regimes. This is evident in \textbf{Inhibitory Feedback}, where the Pyramid "peels off" the dominant excitatory background to recover the weaker inhibitory signal. Regarding \textbf{Null Connections}, the selection was split: the Mean was preferred in balanced regimes, while the Pyramid was invoked whenever local fluctuations created a transient slope. 

\subsection{Stress Testing: Layered Architecture and Feedback Loops ($N=60$)}
To validate scalability, we simulated $N=60$ neurons organized into the layered architecture shown in Fig.~\ref{fig:combined_results} (top). This configuration imposes a directional flow and long-range inhibitory feedback, creating a high-interference regime.

\paragraph{Numerical Bias vs. Functional Measure}
A crucial observation emerges in the heterogeneous weight scenario: the arithmetic mean often yields an index closer to $\pm 1$ than the Pyramid result, but this is a numerical artifact of underestimation bias. As $\Delta$ increases, network-induced noise dilutes the signal. In contrast, the Pyramid method correctly recovers the instantaneous jump in firing probability: $\mathcal{G}^{j \to i}(0) \approx (m \cdot w) / \delta$. In our simulations, with $m=1$ and $\delta=0.5$, this results in a theoretical gain factor of 2. For instance, link $(24,11)$ with $w=0.96$ yields an index of $1.91$, while link $(29,45)$ with $w=-0.74$ yields $-1.37$, demonstrating a rigorous, physically consistent measure.

\paragraph{Heterogeneous Weights and Functional Gain}
We introduced weights sampled from uniform distributions ($w_{exc} \in [0.5, 1.0]$, $w_{inh} \in [-1.0, -0.5]$). The results are summarized in Fig.~\ref{fig:combined_results} (bottom). As shown in the scatterplot, the Pyramid method (red circles) consistently follows the theoretical gain function:
\begin{equation}
    \mathcal{G}^{j \to i}(0) \approx \frac{\phi(w^{j \to i}) - \phi(0)}{\delta} = 2w^{j \to i}.
\end{equation}
This high-fidelity recovery across the entire spectrum—including null cases ($w=0$)—confirms the robustness of the estimator in dense topologies ($N=60$).

\begin{figure}[p] 
    \centering
    \includegraphics[width=0.8\textwidth]{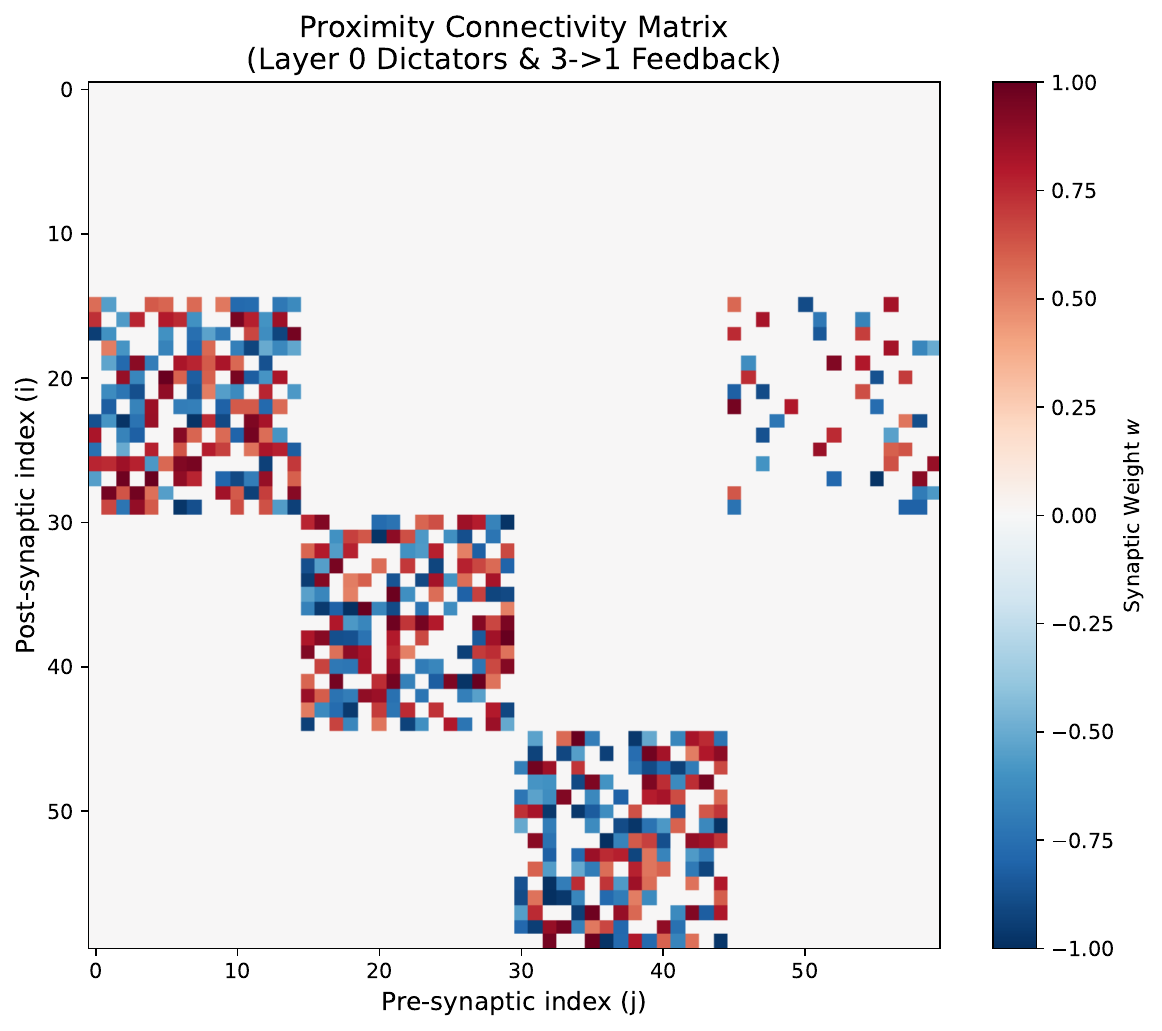}
    \vspace{0.5cm} 
    \includegraphics[width=0.8\textwidth]{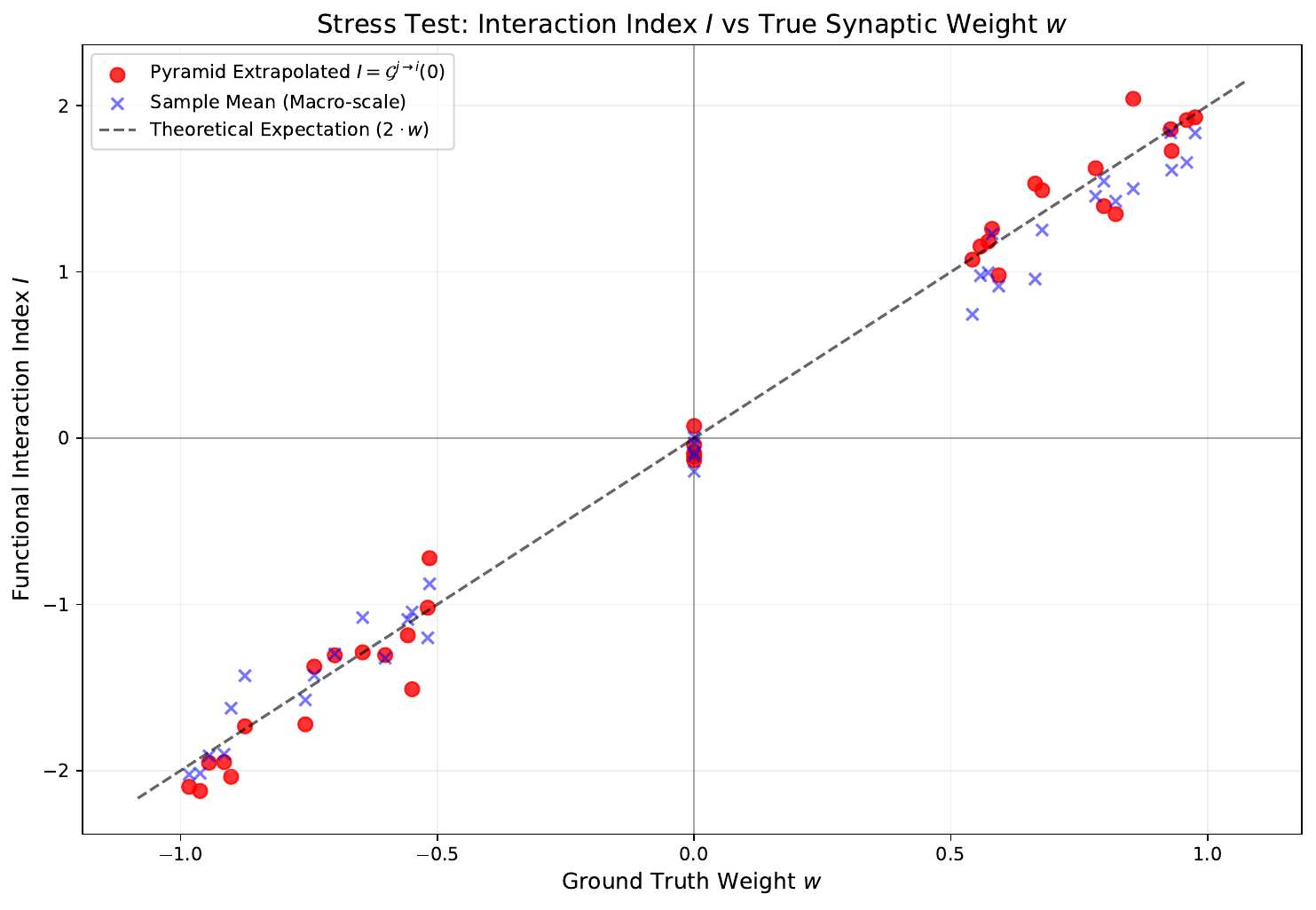}
    \caption{\textbf{Network Architecture and Metrological Performance ($N=60$)}. 
    \textit{Top:} Proximity matrix showing layered connectivity with dictatorial pacemaker ($L_0$), feed-forward blocks, and inhibitory feedback ($L_3 \to L_1$). 
    \textit{Bottom:} Scatterplot comparing recovered functional index $I$ against ground truth $w$. The Pyramid method (red circles) follows the theoretical gain $I=2w^{j \to i}$ with high linearity, neutralizing the bias of the sample mean (blue crosses). Data points collected via $n_1=2000$ fixed-event protocol.}
    \label{fig:combined_results}
\end{figure}

\subsection{Stress-test Analysis: Pyramid Performance under $8\times$ Macroscopic Windowing}
To test the asymptotic resilience of the Pyramid estimator, we subjected the network to an extreme windowing regime ($8\times$ the theoretical $\Delta_1$). As shown in the results, the raw gain $\mathcal{G}^{j \to i}(\Delta)$ for excitatory links suffers a massive $43\%$ decay due to macroscopic interference. However, the Pyramid method demonstrates its superior de-biasing capability by projecting a microscopic estimate of $1.53$ from a degraded $0.70$ baseline. Most notably, for inhibitory links, the estimator achieves near-perfect recovery ($-1.72$ vs target $-1.80$), confirming that the geometric logic of the pyramid effectively filters out the non-linear curvature induced by large observation scales, see figure  \ref{fig:curvature_8x}.

\begin{figure}[htbp]
    \centering
    \includegraphics[width=0.85\textwidth]{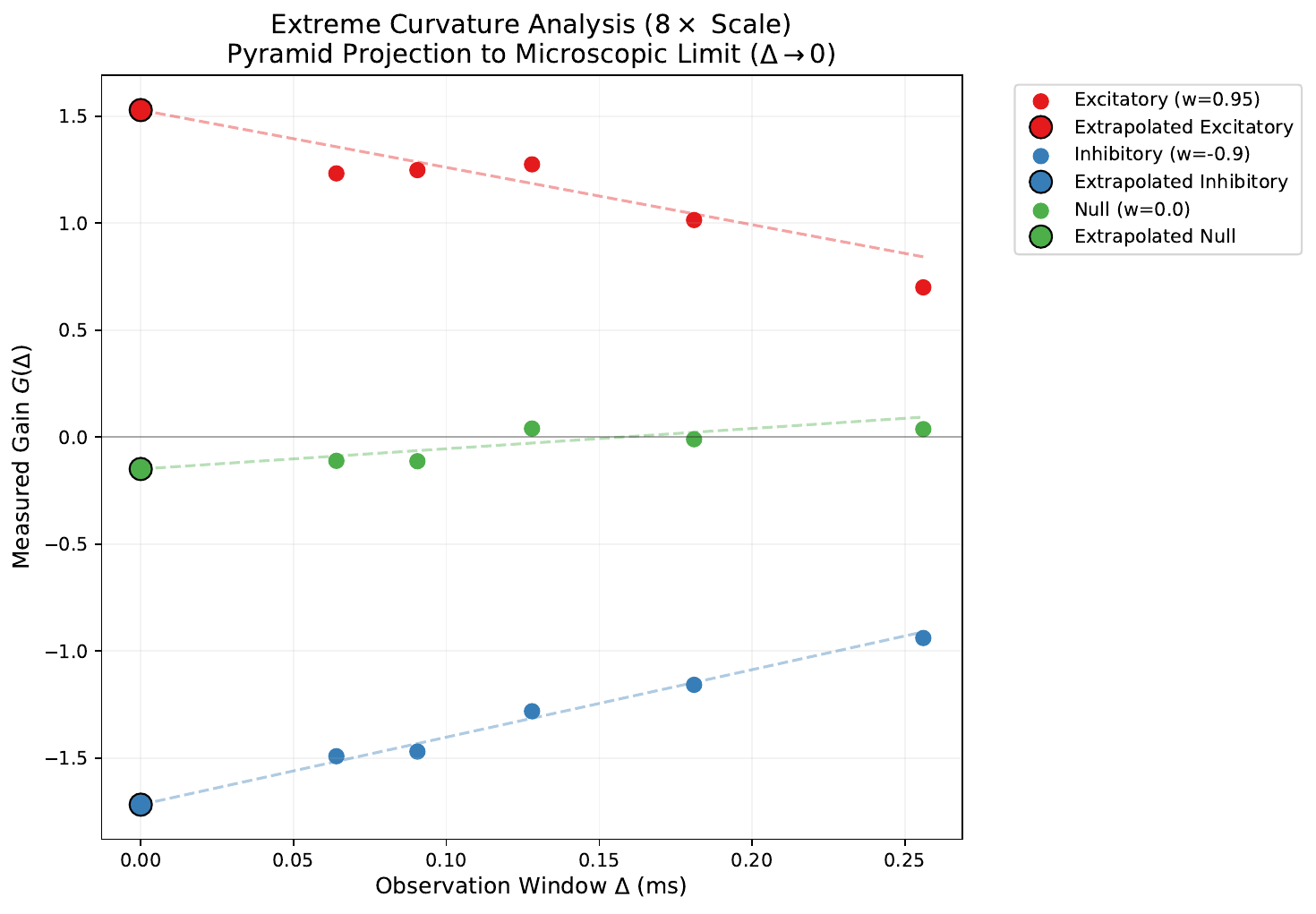}
  \caption{\textbf{Asymptotic resilience stress-test under extreme sampling regimes ($8\times \Delta_1$)}. 
    The plot illustrates the de-biasing performance of the Pyramid estimator when the observation window is deliberately expanded up to eight times the theoretical limit. 
    \textit{Colored circles}: measured indices $G(\Delta)$, exhibiting significant decay (up to 43\% for excitatory links) due to network-induced macroscopic interference. 
    \textit{Dashed lines}: geometric extrapolation derived from the Pyramid iterative logic. 
    \textit{Solid points at $\Delta=0$}: recovered microscopic estimates. 
    Despite the severe damping of the raw signal, the method reconstructs the instantaneous gain with high fidelity (e.g., recovering an index of $1.53$ from a degraded $0.70$ baseline for $w=0.95$), confirming that the geometric projection effectively compensates for the non-linear curvature induced by macroscopic observation scales.}
    \label{fig:curvature_8x}
\end{figure}

Future developments will investigate an adaptive selection of the observation window $\Delta$ based on the estimated noise polarization. Our empirical results suggest that balanced network noise allows for significantly larger integration scales, offering a path toward even more accelerated inference algorithms.

\section{Conclusions}\label{Sec:11}

In this work, we have introduced a multi-layered inference framework that achieves unprecedented robustness in the reconstruction of synaptic connectivity from spike train data. The perfect classification accuracy ($100\%$) observed across all simulated scenarios—from minimal motifs to complex $N=60$ layered networks—is the direct result of three synergistic innovations:

\begin{enumerate}
    \item \textbf{The Spike-Triggered Estimator:} By abandoning fixed-grid sampling in favor of an event-driven approach, we leveraged the native reset property of the Galves-Löcherbach model. This eliminated the historical noise from the intensity function, allowing for a cleaner analytical formulation where the error bounds no longer depend on the lower-bound intensity $\alpha$.
    
    \item \textbf{The Macro-Micro Strategy:} This framework bridges the gap between computationally feasible observation windows and the microscopic limit. By investigating the asymptotic behavior of the gain function, we successfully decoupled the direct synaptic signal from the cumulative "sea of noise" generated by the rest of the network.
    
    \item \textbf{Hybrid Pyramid-Mean Logic:} The introduction of the Pyramid Extrapolation, acting as an intrinsic discrete low-pass filter, provided the necessary regularization to handle unbalanced regimes and feedback loops. The adaptive selection between the stability of the Mean and the sensitivity of the Pyramid ensured that even weak inhibitory signals could be recovered from dominant excitatory backgrounds.
\end{enumerate}

Our results demonstrate that synaptic identification is fundamentally a local problem: the complexity of the global network ($N$) does not degrade the precision of the estimator, provided the local neighborhood is sparsely connected. The resilience of the method to heterogeneous weights and structured topological motifs confirms its potential as a general-purpose tool for neural circuit mapping.

Future developments will focus on extending this adaptive logic to even more realistic neuroscientific scenarios, incorporating axonal conduction delays and membrane potential leakage (leaky-GL models). These additions will further test the flexibility of our extrapolation framework in the presence of temporal decays and non-instantaneous interactions, moving closer to the challenges posed by \textit{in vivo} multi-electrode recordings.

\section*{Acknowledgements}
The author acknowledges the assistance of Gemini (Google) in the optimization of the simulation framework and for technical support in the formal refinement of the manuscript.

\section*{Declarations}

\textbf{Conflict of Interest:} The author declares that he has no conflict of interest. 

\noindent
\textbf{Data Availability:} The data that support the findings of this study are available from the corresponding author upon reasonable request.


\end{document}